\documentclass{SciPost}

\hypersetup{
    colorlinks,
    linkcolor={red!50!black},
    citecolor={blue!50!black},
    urlcolor={blue!80!black}
}

\usepackage[bitstream-charter]{mathdesign}
\usepackage{braket}
\usepackage{amsthm}
\usepackage{algorithm}  
\usepackage{algorithmic} 

\newtheorem{definition}{Definition}
\newtheorem{lemma}{Lemma}
\makeatletter

\newenvironment{breakablealgorithm}
{
		\begin{center}
			\refstepcounter{algorithm}
			\hrule height.8pt depth0pt \kern2pt
			\renewcommand{\caption}[2][\relax]{
				{\raggedright\textbf{\ALG@name~\thealgorithm} ##2\par}%
				\ifx\relax##1\relax 
				\addcontentsline{loa}{algorithm}{\protect\numberline{\thealgorithm}##2}%
				\else 
				\addcontentsline{loa}{algorithm}{\protect\numberline{\thealgorithm}##1}%
				\fi
				\kern2pt\hrule\kern2pt
			}
		}{
		\kern2pt\hrule\relax
	\end{center}
}
\makeatother

\DeclareSymbolFont{usualmathcal}{OMS}{cmsy}{m}{n}
\DeclareSymbolFontAlphabet{\mathcal}{usualmathcal}

\fancypagestyle{SPstyle}{
\fancyhf{}
\lhead{\colorbox{scipostblue}{\bf \color{white} ~SciPost Physics }}
\rhead{{\bf \color{scipostdeepblue} ~Submission }}

\fancyfoot[C]{\textbf{\thepage}}
}

\begin{document}

\pagestyle{SPstyle}

\begin{center}{\Large \textbf{\color{scipostdeepblue}{
Improved belief propagation decoding algorithm based on decoupling representation of Pauli operators for quantum LDPC codes\\
}}}\end{center}

\begin{center}\textbf{
Zhengzhong Yi\textsuperscript{1},
Zhipeng Liang\textsuperscript{1},
Kaixin Zhong\textsuperscript{1},
Yulin Wu\textsuperscript{1},
Zhou Fang\textsuperscript{1} and
Xuan Wang\textsuperscript{1$\star$}
}\end{center}

\begin{center}
{\bf 1} Harbin Institute of Technology (Shenzhen), Shenzhen, China
\\[\baselineskip]
$\star$ \href{mailto:email1}{\small wangxuan@cs.hitsz.edu.cn}
\end{center}

\section*{\color{scipostdeepblue}{Abstract}}
\boldmath\textbf{%
We propose a new method called decoupling representation to represent Pauli operators as vectors over $GF(2)$, based on which we propose partially decoupled belief propagation and fully decoupled belief propagation decoding algorithm for quantum low density parity-check codes. These two algorithms have the capability to deal with the correlations between the $X$ part and the $Z$ part of the vectors in symplectic representation, which are introduced by Pauli $Y$ errors. Hence, they can not only apply to CSS codes, but also to non-CSS codes. Under the assumption that there is no measurement error, compared with traditional belief propagation algorithm in symplectic representation over $GF(2)$, within the same number of iterations, the decoding accuracy of partially decoupled belief propagation and fully decoupled belief propagation algorithm is significantly improved in pure $Y$ noise and depolarizing noise, which supports that decoding algorithms of quantum error correcting codes might have better performance in decoupling representation than in symplectic representation. The impressive performance of fully decoupled belief propagation algorithm might promote the realization of quantum error correcting codes in engineering.
}

\vspace{\baselineskip}

\noindent\textcolor{white!90!black}{%
\fbox{\parbox{0.975\linewidth}{%
\textcolor{white!40!black}{\begin{tabular}{lr}%
  \begin{minipage}{0.6\textwidth}%
    {\small Copyright attribution to authors. \newline
    This work is a submission to SciPost Physics. \newline
    License information to appear upon publication. \newline
    Publication information to appear upon publication.}
  \end{minipage} & \begin{minipage}{0.4\textwidth}
    {\small Received Date \newline Accepted Date \newline Published Date}%
  \end{minipage}
\end{tabular}}
}}
}


\vspace{10pt}
\noindent\rule{\textwidth}{1pt}
\tableofcontents
\noindent\rule{\textwidth}{1pt}
\vspace{10pt}


\section{Introduction}
\label{1}
Quantum computing shows a potential to solve certain problems which classical computing might not solve within an acceptable time and resources\cite{365700,arute2019quantum,zhong2020quantum}.  However, in current quantum computers, qubits are easily disturbed by environmental noise and quantum operations on qubits are not accurate enough\cite{wang2017single}, which will destroy the quantum information carried by the systems. This results that current quantum computers have difficulties to fulfil the potential of quantum computing. Fortunately, in 1996, Shor  proposed that we can protect quantum information in an imperfect quantum system by quantum error correction codes (QECCs)\cite{PhysRevA.52.R2493}. 

Stabillizer codes\cite{gottesman1997stabilizer,calderbank1998quantum} are an important and promising family of QECCs. In recent years, among the stabilizer codes, quantum low density parity-check (QLDPC) codes\cite{mackay2004sparse,babar2015fifteen,gottesman2013fault,tillich2013quantum,bombin2006topological,breuckmann2021quantum,xu2023tailored,roffe2022bias,bravyi1998quantum,freedman2001projective,fowler2012surface,kitaev1997quantum,kitaev2003fault,bonilla2021xzzx}, such as planar surface codes\cite{bravyi1998quantum,freedman2001projective,fowler2012surface}, toric codes\cite{kitaev1997quantum,kitaev2003fault} and XZZX surface codes\cite{bonilla2021xzzx}, tend to attract more and more attentions\cite{google2021exponential,google2023suppressing,marques2022logical}  due to the following two reasons. First, their stabilizer generators have low weight, which means that the number of qubits involved by each generator is small. In general, the less qubits involved, the more accurate the measurement of the generators is. Besides, low weight stabilizer generators are the prerequisite for a QECC having good operation locality, which will also benefit the accuracy of the stabilizer generator measurement. Second, the number of stabilizer generators involved by a certain qubit is finite. These features can simplify decoding process. For instance, for a QLDPC code, if the number of its stabilizer generators involved by a certain qubit is less than 2, we can use minimum-weight perfect matching (MWPM)\cite{edmonds1965paths} algorithm to decode it. Otherwise, the decoding problem is a hypergraph matching problem\cite{delfosse2014decoding,wang2009graphical} and there is still no efficient algorithm to deal with it\cite{zhang2021quantum}.

Typical decoding algorithms for QLDPC codes include MWPM algorithm\cite{edmonds1965paths}, tensor network (TN)\cite{bravyi2014efficient} algorithm, renormalization group (RG)\cite{duclos2010fast} algorithm and belief propagation (BP)\cite{babar2015fifteen,roffe2020decoding,kuo2020refined,kuo2020refined1,lai2021log,kuo2022exploiting,poulin2008iterative} algorithm. Among these algorithms, MWPM, TN and RG algorithms ask for the codes having a clear geometric layout,  while BP algorithm does not ask for this. 

BP decoding algorithm is an iterative algorithm and the time complexity of each iteration is $O(n)$. Each round of BP decoding is composed of multiple iterations  and at each iteration it will give out an estimated error. If the error satisfies the check equations, we say BP converges and BP returns this error as decoding result. Otherwise, BP will restart a new iteration until it converges or the maximum number of iterations is reached. If BP cannot converge after all iterations, we say BP “fails” in this round. In practice, the maximum iteration number of BP can always be chosen as a large enough constant $T_{max}$ or to be proportional to $loglogn$\cite{mackay1999good}, thus the time complexity of BP is $O(nT_{max})$ or $O(nloglogn)$, where $n$ is the code length, and this low complexity makes it a promising algorithm for real-time error correcting scenario. 

According to the finite field of the codes decoded by BP, the application of BP can be divided into two classes — quaternary BP\cite{babar2015fifteen,kuo2022exploiting,kuo2020refined,kuo2020refined1,poulin2008iterative} which decodes codes over $GF(4)$ and binary BP\cite{babar2015fifteen,roffe2020decoding} which decodes codes over $GF(2)$. Due to the existence of three Pauli error operators and the identity operator, QECCs are naturally codes over $GF(4)$. However, in most cases, researchers are used to deal with these QECCs over $GF(2)$, using the method of representing parity-check matrices and Pauli operators shown in Eq. (10.83) in Ref. \cite{nielsen2002quantum}, namely, the \textbf{\textit{symplectic representation}}. 

However, both quaternary BP and binary BP faces with some specific difficulties.

For quaternary BP, since the Tanner graph of a quantum stabilizer code contains many 4-cycles, its decoding convergence rate of a single round is low\cite{babar2015fifteen,kuo2020refined,kuo2020refined1,poulin2008iterative}. To improve the performance of quaternary BP, some methods are proposed. For example, In Ref. \cite{kuo2022exploiting}, the researchers exploit the degeneracy of QLDPC codes and propose the method of performing multiple rounds of BP decoding, while the number of rounds is influenced by a set of parameters related to decoding. However, these parameters should be chosen precisely in specific situation, which is difficult to do in practice\cite{kuo2022exploiting}. In addition, how many rounds are needed is not clear, which makes the total time complexity of the whole decoding process still unknown. Another example is using quaternary ordered statistics decoding (OSD)\cite{Panteleev2021degeneratequantum,kung2023belief} as post-processing to  improve the performance of quaternary BP when it fails to converge.

For binary BP, the symplectic representation of Pauli operators can reduce 4-cycles on the Tanner graph of the binary parity-check matrix of the corresponding quantum stabilizer code. However, this representation results that Pauli $Y$ errors introduce the correlations between vectors $\textbf{e}_x$ and $\textbf{e}_z$ (which we will explain in Sect. \ref{3.1}) that represent Pauli operators, which will decrease the performance of binary BP. To deal with this problem, the binary BP using the X/Z correlations for CSS codes is proposed\cite{6874997}. However, for non-CSS codes, especially in $Y$-biased noise, even combined with OSD, the error correction performance of binary BP in symplectic representation is unsatisfactory.

In this paper, we focus on binary BP decoding, and improve its performance by eliminating the correlations introduced by Pauli $Y$ errors, which can apply  to both non-CSS and CSS codes, especially in $Y$-biased noise. This is why we perform simulations on XZZX surface codes and planar surface codes, which are well-known non-CSS and CSS codes, respectively. 

In symplectic representation, the single-qubit Pauli operators are represented by two bits. Hence, for a $[[n,k]]$ quantum stabilizer code, where $n$ is the code length and $k$ is the number of logical qubits, the Pauli errors acting on $n$ qubits are represented by a binary vector with size of $2n$, and its parity-check matrix can be represented by a binary matrix with dimension of $(n-k)\times 2n$.
In this paper, we propose a new general method to represent Pauli operators and the parity-check matrices of QECCs. Different from the symplectic representation, this new method represents single-qubit Pauli operators by three bits. In this new representation, for a $[[n,k]]$ quantum stabilizer code, the Pauli errors acting on $n$ qubits are represented by binary vectors with size of $3n$, and its parity-check matrix can be represented by a binary matrix with dimension of $(n-k)\times 3n$. By this method, since the representations of Pauli $X$, $Y$ and $Z$ errors seem to be decoupled from each other (we will explain this in Sect. \ref{3}), we call this method the \textbf{\textit{decoupling representation}}.  Because of this feature, we conjecture that decoding algorithms  tend to have more balanced capability to decode Pauli $X$, $Y$ and $Z$ errors, which might lead to a better decoding performance under certain noise channels, such as $Y$-biased noise channel and depolarizing noise channel where all the three Pauli errors might happen.  

Based on decoupling representation, we modified binary BP algorithm for QLDPC codes. In decoupling representation, the Tanner graph, message update  rules and hard decision need modifying. We first only modify the Tanner graph and obtain \textbf{\textit{partially decoupled binary BP}} (\textbf{\textit{PDBP}}). Then we perform simulations to compare the performance of PDBP and SBP (in convenience, we call binary BP in symplectic representation the  \textbf{\textit{symplectic BP}} or  \textbf{\textit{SBP}}). All the simulations in this paper are combined with order$-0$ OSD and under the assumption that there is no measurement error in syndrome detection. Simulation results show that, for XZZX surface code and planar surface code, the capability of PDBP to decode Pauli $X$ and $Z$ error is the same as that of SBP. Besides, PDBP shows the balanced capability to decode Pauli $X$, $Y$, $Z$ errors, while SBP shows much stronger capability to decode Puali $X$, $Z$ errors than Pauli $Y$ errors, which means that PDBP has stronger capability to decode Pauli $Y$ errors. In consequence,  under the depolarizing noise channel, the decoding accuracy of PDBP is higher than that of SBP. This result supports our conjecture about the performance of decoding algorithms in decoupling representation, which is mentioned in the previous paragraph.

To further improve the decoding accuracy, based on PDBP we modify the rules of message update and hard decision, and obtain  \textbf{\textit{fully decoupled binary BP}} (\textbf{\textit{FDBP}}). These modifications make the message update and hard decision compatible with the decoupling representation which requires that in the binary vector with size of $3n$ which represents a Pauli operator acting on $n$ qubits, no more than one 1 can show up in the same positions corresponding to the same qubit. Simulation results show that FDBP performs even better than PDBP in the sense of decoding accuracy. Moreover, the code-capacity thresholds of FDBP on XZZX surface code and planar surface code are more close to theoretical values. The impressive performance of FDBP might promote the realization of QECC in engineering.

The rest of paper is organized as follows. In Sect. \ref{2}, we introduce some preliminaries, including quantum stabilizer codes, QLDPC codes, Tanner graph, the fundamentals of BP algorithm and ordered statistics decoding (OSD). In Sect. \ref{3}, we introduce the idea of decoupling representation of Pauli operators, decoupled parity-check matrix, FDBP and PDBP. The simulation results are presented in Sect. \ref{4}. In Sect. \ref{5}, we conclude our work and analyze the existing problems.

\section {Preliminaries}
\label{2}

In this section, we briefly introduce some preliminaries, including stabilizer codes, QLDPC codes, Tanner graph, the fundamentals of BP algorithm and OSD.
\subsection {Quantum stabilizer codes}
\label{2.1}
Quantum stabilizer codes are an important and promising class of QECCs, whose code construction is analogous to classical linear codes.

Quantum stabilizer codes are based on the concept of a subgroup of the Pauli group, which is called stabilizer group. The Pauli group $\mathcal{G}_1$ on a single qubit is a set which consists of all Pauli operators together with multiplicative constants $\{\pm1,\pm i\}$, namely, $\mathcal{G}_1=\{\pm I,\pm iI,\pm X,\pm iX,\pm Y,\\
\pm iY,\pm Z,\pm iZ\}$. The general Pauli group $\mathcal{G}_n$ on $n$ qubits is $n$-fold tensor product of $\mathcal{G}_1$, namely, $\mathcal{G}_n=\mathcal{G}_1^{\otimes n}$. The weight of an operator $P\in \mathcal{G}_n$ is defined as the number of qubits on which it acts nontrivially, and we use notation $wt(\cdot)$ to denote it. For instance, $wt(I_1 X_2 Y_3 Z_4 )=3$.

An $[[n,k]]$ quantum stabilizer code $C$ uses $n$ physical qubits to encode $k$ logical qubits into a subspace $Q_C$ of $(\mathbb{C}^2 )^{\otimes n}$, namely, the code space. In fact, the space $Q_C$ is the common +1-eigenspace of a set of independent commuting operators $S_1,\cdots,S_{n-k}\in \mathcal{G}_n$, which are called stabilizer generators. More precisely,
\begin{equation}
	Q_C=\{\ket{\varphi}\in (\mathbb{C}^2 )^{\otimes n}:S\ket{\varphi}=\ket{\varphi},\forall S\in \mathcal{S}\}
\end{equation}
where $\mathcal{S}$ is the stabilizer group of code $C$ generated by $S_1,\cdots,S_{n-k}$, namely, $\mathcal{S}=\langle S_1,\cdots,S_{n-k} \rangle$. We can see that giving a set of stabilizer generators  $\langle S_1,\cdots,S_{n-k}\rangle$  of code $C$ is equivalent to explicitly giving the code space $Q_C$.

The error syndrome $\textbf{s}=(s_1,\cdots,s_{n-k})$ of an error $E\in \mathcal{G}_n$ is a binary vector, where $s_i=1$ if $E$ anti-commutes with $S_i$ and 0 otherwise.

In quantum information theory, researchers usually map Pauli $I$, $X$, $Y$, $Z$ operators onto $GF(4)$ or $GF(2)$. Intuitively, these 4 Pauli operators can be individually represented by 4 elements in $GF(4)$, namely,

\begin{equation}
	I\rightarrow0,\ X\rightarrow1,\ Z\rightarrow \omega,\ Y\rightarrow \bar{\omega}
\end{equation}
where $0,\ 1,\ \omega,\ \bar{\omega}\in GF(4)$.

Based on this Pauli-to-$GF(4)$ isomorphism, the $(n-k)$ stabilizer generators of a $[[n,k]]$ quantum stabilizer code make up the rows of the parity-check matrix over $GF(4)$ with dimension of $(n-k)\times n$.

Besides, the $I$, $X$, $Y$, $Z$ Pauli operators can also be mapped onto $GF(2)^2$, namely,
\begin{equation}
	I\rightarrow(0,0),\ X\rightarrow(1,0),\ Z\rightarrow (0,1),\ Y\rightarrow (1,1)
\end{equation}

Based on this Pauli-to-$GF(2)$ isomorphism, the binary parity-check matrix $H$ of a $[[n,k]]$ quantum stabilizer code is a block matrix with dimension $(n-k)\times 2n$, which consists of two $(n-k)\times n$ binary matrices $H_x$ and $H_z$, namely,
\begin{equation}
	H=(H_x\mid H_z)
\end{equation}

In the same way, any operator $E\in \mathcal{G}_n$ acting on $n$ qubits can be represented as a binary vector $\textbf{e}=(\textbf{e}_x\mid \textbf{e}_z )$, then the syndrome $\textbf{s}$ of $E$ is computed by
\begin{equation}
	\label{syndrome}
	s=(H_x\cdot \textbf{e}_z + H_z\cdot \textbf{e}_x)\ mod\ 2
\end{equation}

\subsection {Quantum low density parity-check codes}
\label{2.2}
Quantum low density parity-check (QLDPC) codes are the analogue of classical low density parity-check (LDPC) codes \cite{gallager1962low,mackay1999good} in quantum fields and a promising class of quantum stabilizer codes.

A LDPC code can be characterized by two constants $l_r$ and $l_c$, which represents the upper bound of row weight and column weight of its parity-check matrix respectively. The name LDPC comes from the feature of these codes that their parity-check matrices are sparse. Similarly, a QLDPC code can also be characterized by two constants $w_r$ and $w_c$, which represents the upper bound of row weight and column weight of its parity-check matrix over $GF(4)$, respectively. Since the rows of parity-check matrix correspond to the stabilizer generators, and columns correspond to qubits, one can also interpret $w_r$ as the upper bound of the weight of each stabilizer generator and $w_c$ as the upper bound of the number of stabilizer generators involved with a certain qubit. For instance, the $w_r$ and $w_c$ of toric code are 4 and 4 respectively. Same as LDPC codes, the parity-check matrices of QLDPC codes are also sparse.

\subsection {Tanner graph, belief propagation, and ordered statistics decoding}
\label{2.3}

\subsubsection {Tanner graph}
\label{2.3.1}
The parity-check matrix $H$ of an error correction code (both for quantum codes and classical codes) can be depicted by a bipartite graph $G=(V,C,E)$ called Tanner graph. Assume the dimension of $H$ is $m\times n$, then $V=\{v_1,\cdots,v_n\}$ is the set of $n$ variable nodes, $C=\{c_1,\cdots,c_m\}$ is the set of $m$ check nodes and $E=\{(c_i,v_j ),\ 1\le i\le m,\ 1\le j\le n\}$ is the set of edges that connect variable nodes $v_j$  and check nodes $c_i$  if $H_{ij}\neq0$ for all $1\le i\le m$, $1\le j\le n$. Here, we use notations $\mathcal{N}(c_i)$ to denote the set of neighboring variable nodes of the check node $c_i$ and $\mathcal{N}(v_j)$ the set of neighboring check nodes of the variable node $v_j$. If there is an edge between $c_i$ and $v_j$, then $v_j$ is a neighboring variable node of $c_i$, and $c_i$ is a neighboring check node of $v_j$. We use notations $\mathcal{N}(c_i)\slash v_{j^{\prime}}$ and $\mathcal{N}(v_j )\slash c_{i^{\prime}}$ to denote the set of neighboring variable nodes of the check node $c_i$ excluding $v_{j^{\prime}}$ and the set of neighboring check nodes of the variable node $v_j$ excluding $c_{i^{\prime}}$, respectively.

In Fig. \ref{fig:Tanner graph}, we show two examples of quantum stabilizer code’s Tanner graph. Fig. \ref{fig:Tanner graph}(a) shows the Tanner graph of a quantum stabilizer code with parity-check matrix $H_1$ over $GF(4)$. Fig. \ref{fig:Tanner graph}(b) shows the corresponding Tanner graph of the same code whose parity-check matrix is transformed to $H_2$ over $GF(2)$ in symplectic representation. One can find that there are two classes of variable nodes in Fig. \ref{fig:Tanner graph}(b). One class corresponds to vector $\textbf{e}_x$ and the other corresponds to vector $\textbf{e}_z$.
\begin{figure*}[htbp]
	\centering
	\includegraphics[width=1\textwidth]{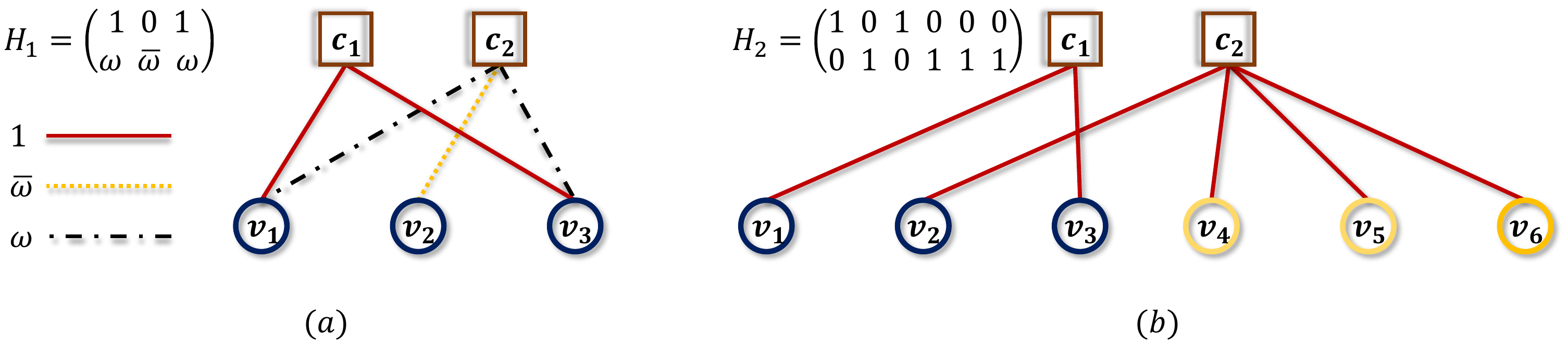}
	\caption{\textbf{Two kinds of Tanner graph of a quantum stabilizer code $C$ whose stabilizer generators are $S_1=XIX$ and $S_2=ZYZ$.} (a) The Tanner graph of parity-check matrix $H_1=\left(\begin{matrix} 1 & 0 & 1 \\ \omega  & \bar{\omega} &\omega \end{matrix}\right)$ over $GF(4)$ corresponding to code $C$. (b) The Tanner graph of parity-check matrix $H_2=\left(\begin{matrix} 1 & 0 & 1 & 0 & 0 & 0\\0 & 1 &0 & 1& 1 & 1\end{matrix}\right)$ over $GF(2)$ in symplectic representation corresponding to code $C$. Variable nodes $v_1$, $v_2$ and $v_3$ corresponds to vector $\textbf{e}_z$ and variable nodes $v_4$, $v_5$ and $v_6$ corresponds to vector $\textbf{e}_x$.}
	\label{fig:Tanner graph}
\end{figure*}

\subsubsection {Belief propagation}
\label{2.3.2}
Given a syndrome $\textbf{s}=(s_1,\cdots,s_{n-k})$ of a $[[n,k]]$ QLDPC code, the decoding algorithms are supposed to find out the most likely error $\hat{E}$ whose syndrome is consistent with $\textbf{s}$, namely,
\begin{equation}
	\hat{E}=\mathop{\arg\max}\limits_{E}\ P(E|\textbf{s})
\end{equation}
However, directly computing $P(P|\textbf{s})$ for each possible value of $E$ is a $NP$-complete problem as the code length of the QLDPC code tends to infinite\cite{hsieh2011np,iyer2015hardness}. Fortunately, BP provides us a method with linear time complexity to efficiently (when there are short cycles in the Tanner graph, also approximately) compute it\cite{wiberg1996codes,mceliece1998turbo,kschischang2001factor}. For a Pauli error $E=(E_1,\cdots,E_n)$ acting on $n$ qubits $(E_i\in\{I,X,Y,Z\},1\le i\le n)$ with the corresponding syndrome $\textbf{s}$, BP gives out the most likely error on $i$th qubit $\hat{E}_i$ namely
\begin{equation}
	\hat{E}_i=\mathop{\arg\max}\limits_{E_i}\ P(E_i|\textbf{s})
\end{equation}
where $P(E_i|\textbf{s})$ is the marginalized probability of the error $E_i$ acting on the $i$th qubit. BP computes this probability by iteratively passing messages on the edge between the check nodes and the variable nodes of the Tanner graph. Within each iteration, message passing algorithm is used to accelerate the calculation of marginal probability. The calculation results of each iteration will be taken as the input of next iteration to improve the accuracy of the estimation of final marginal probability. There are multiple specific algorithms which can be used to update the messages on the edges in message passing algorithm, including sum-product algorithm \cite{mackay1999good}, min-sum algorithm \cite{fossorier1999reduced} and modified versions of min-sum algorithm\cite{emran2014simplified,chen2002near,chen2002density}. In our work, we have tried sum-product algorithm and min-sum algorithm.

Next, we briefly show that how binary log-domain BP works and focus on binary BP in the rest of paper. Readers can see \cite{babar2015fifteen,kuo2022exploiting,kuo2020refined,kuo2020refined1,poulin2008iterative} for more detail of quaternary BP.

Given a $m\times n$ binary parity-check matrix $H$, syndrome $\textbf{s}=(s_1,\cdots,s_m)$, the error rate on each variable nodes $p=(p_1,\cdots,p_n)$, and the maximum number of iterations $iter_{max}$, the decoding procedure of BP can be divided into the following steps:

\noindent\textbf{Step 1: Initialization}

For each $1\le j\le n$, compute
\begin{equation}
	\gamma_j=\ln \frac{p(E_j=0)}{p(E_j=1)} = \ln\frac{1-p_j}{p_j}
\end{equation}
which is the priori information provided to the decoder. The variable-to-check  message $m_{v_j\Rightarrow c_i}$ sent by variable node $v_j$ to its neighboring check node $c_i\in \mathcal{N}(v_j)$ is initialized as
\begin{equation}
	m_{v_j\Rightarrow c_i}=\gamma_j
\end{equation}

\noindent\textbf{Step 2: Horizontal update}

The check-to-variable message $m_{c_i\Rightarrow v_j}$ sent by check node $c_i$ to its neighboring variable node $v_j\in \mathcal{N}(c_i)$ is updated as a function of the variable-to-check  messages $m_{v_{j^{\prime}}\Rightarrow c_i }$ previously received by the check node $c_i$ from its neighboring variable nodes $v_{j^{\prime}}\in \mathcal{N}(c_j )/v_j$, as shown in Fig. \ref{fig:message update}(a).

\noindent\textbf{Step 3: Vertical update}

The variable-to-check  message $m_{v_j\Rightarrow c_i}$ sent by variable node $v_j$ to its neighboring check node $c_i\in \mathcal{N}(v_j)$ is updated as a function of the check-to-variable messages $m_{c_{i^{\prime}}\Rightarrow v_j}$ previously received by the variable node $v_j$ from its neighboring check nodes $c_{i^{\prime}}\in \mathcal{N}(v_j )/c_i$, as shown in Fig. \ref{fig:message update}(b).

\noindent\textbf{Step 4: Hard decision}

Finally, for each $1\le j\le n$, a posteriori information $\tilde{\gamma}_j$ is computed as a function of the priori information $\gamma_j$ and the check-to-variable message $m_{c_i\Rightarrow v_j}$ received by the variable node $v_j$ from its all neighboring check nodes $c_i\in \mathcal{N}(v_j)$, namely,
\begin{equation}
	\tilde{\gamma}_j=\gamma_j + \sum_{c_j\in \mathcal{N}(v_j)} m_{c_i\Rightarrow v_j}
\end{equation}
as shown in Fig. \ref{fig:message update}(c). Let $\hat{E}_i = 1$ if $\tilde{\gamma}_j<0$ and 0 otherwise such that we obtain an estimated error vector $\hat{E}=(\hat{E}_1,\cdots,\hat{E}_n)$. If the syndrome $\hat{\textbf{s}}$ of $\hat{E}$ is consistent with the input syndrome $\textbf{s}$, then the procedure halts and return “converge”. Otherwise, the algorithm repeats from Horizontal update step until $\hat{\textbf{s}}$ is consistent with $\textbf{s}$ or the maximum number of iterations $iter_{max}$ is reached.

\begin{figure*}[htbp]
	\centering
	\includegraphics[width=0.7\textwidth]{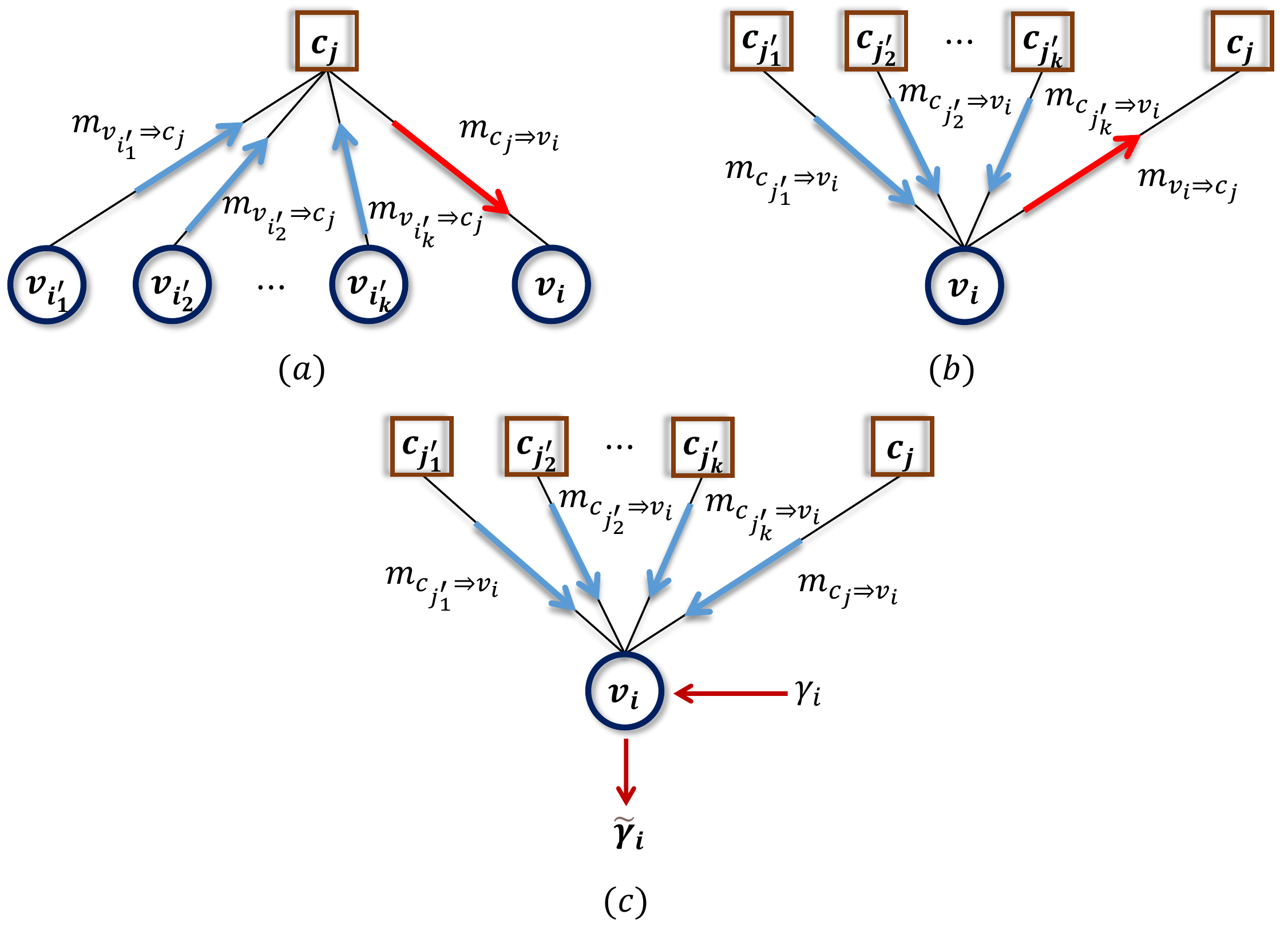}
	\caption{\textbf{Message update on Tanner graph.} (a) Horizontal update (b) Vertical update (c) Hard decision}
	\label{fig:message update}
\end{figure*}

\subsubsection {Ordered statistics decoding}
\label{2.3.3}

Ordered statistics decoding (OSD) is a post-processing algorithm which can be used when BP fails. OSD was originally proposed to improve the waterfall performance in classical LDPC codes\cite{fossorier1995soft}. For a classical LDPC code with parity-check matrix $H^c$, the decoding problem is to find out an error $\hat{e}$ such that $H^c\hat{e}=\textbf{s}$, where $\textbf{s}$ is a syndrome. Notice that $H^c$ is row full rank and the number of rows is smaller than the number of columns. Thus, there is no unique solution for $\hat{e}$. However, it is possible to find out a set of linearly independent columns $J$ of $H^c$ which constitutes a full rank matrix $H_J^c$ such that $\hat{e}={H_J^c}^{-1} \textbf{s}$. Different choice of $J$ leads to different solution $\hat{e}$ and expect to find out the low-weight solution. When BP fails to converge, the last round of posteriori information $\tilde{\pmb{\gamma}}$ still tells us which bits may be flipped with high probability, which is the key idea of OSD. It has been proved that OSD has good performance for some QLDPC codes \cite{roffe2020decoding,Panteleev2021degeneratequantum,kung2023belief}.

Next, we only briefly describe the order$-0$ OSD, which is used in this paper. The performance of BP combined with order$-0$ OSD is good enough in many cases and it is the simplest OSD decoder. Readers can see \cite{roffe2020decoding,Panteleev2021degeneratequantum,kung2023belief} for more detail of high order OSD.

Given a $m\times n$ binary parity-check matrix $H$, syndrome $\textbf{s}=(s_1,\cdots,s_m)$ and the last round of posteriori information $\tilde{\pmb{\gamma}}=(\tilde{\gamma}_1,\cdots,\tilde{\gamma}_n )$, the procedure of order$-0$ OSD can be divided into the following steps:

\noindent\textbf{Step 1: Ranking bit indices.}

The last round of posteriori information $\tilde{\pmb{\gamma}}=(\tilde{\gamma}_1,\cdots,\tilde{\gamma}_n )$ represents the reliability of each bit. Rank them in ascending order according to $\tilde{\pmb{\gamma}}$ and we obtain a ranked list of bit indices $L_R$ (notice that the smaller $\tilde{\gamma}_j$ is, the higher the probability of $j$th bit being flipped is).

\noindent\textbf{Step 2: Rearranging $H$.}

Rearrange the columns of $H$ according to $L_R$ and we obtain an ordered matrix $H_{L_R}$.

\noindent\textbf{Step 3: Finding out a full rank matrix $H_{L_{R}[m]}$.}

Select the first $m$ linearly independent columns $L_R[m]$ of $H_{L_R}$ to constitute a full rank matrix $H_{L_{R}[m]}$.

\noindent\textbf{Step 4: Solving the nonhomogeneous linear equation.}

Solve the nonhomogeneous linear equation $H_{L_{R}[m]} \hat{\textbf{e}}=\textbf{s}$ and we obtain a unique solution $\hat{\textbf{e}}_{L_R[m]}= {H_{L_{R}[m]}}^{-1}\textbf{s}$.

\noindent\textbf{Step 5: Outputting decoding result.}

Let $\hat{\textbf{e}}_{L_R[r]}=\textbf{0}$ be an all-zero vector with size of $(n-m)$, where $L_R[r]$ represents the remaining set of $L_R$ excluding the set $L_R[m]$. Thus, the order$-0$ OSD decoding result is $\hat{\textbf{e}}_{L_R}=(\hat{\textbf{e}}_{L_R[m]},\hat{\textbf{e}}_{L_R[r]})=(\hat{\textbf{e}}_{L_R[m]},\textbf{0})$. Last, map $\hat{\textbf{e}}_{L_R}$ to the original bit ordering and obtain decoding result $\hat{\textbf{e}}_{L_R}\rightarrow \hat{\textbf{e}}$ such that $H\hat{\textbf{e}}=\textbf{s}$.

\section {Belief propagation decoding based on decoupling representation of Pauli operators}
\label{3}

In this section, we first introduce the idea of decoupling representation of Pauli operators, based on which we give out the decoupled parity-check matrix of quantum stabilizer codes and propose PDBP and FDBP for QLDPC codes.

It must be emphasized that in this paper, we only consider the general independent single-qubit Pauli noise channel in which each qubit independently suffers a Pauli error $I$, $X$, $Y$, or $Z$ with error probability $1-p$, $p_X$, $p_Y$, $p_Z$, respectively, where $p=p_X+p_Y+p_Z$ is physical qubit error rate. If $p_X=p_Z=p_Y$, the channel is a depolarizing noise channel. If $p_X\neq0$, $p_Y=0$, $p_Z=0$, the channel is a pure $X$ noise channel, and the similar for pure $Z$ noise channel and pure $Y$ noise channel.

\subsection {Decoupling representation of Pauli operators and decoupled parity-check matrix}
\label{3.1}
In symplectic representation, any Pauli error $E$ acting on $n$ qubits can be represented as a binary vector $\textbf{e}=(\textbf{e}_x\mid \textbf{e}_z )$, where $\textbf{e}_x$ and $\textbf{e}_z$ are both binary vectors with size of $n$. Actually, $\textbf{e}_x$ consists of two error constituents — Pauli $X$ errors and Pauli $Y$ errors. If we denote these two error constituents by vectors $\textbf{e}_x^{\prime}$ and $\textbf{e}_y^{\prime}$, respectively, then
\begin{equation}
	\label{ex}
	\textbf{e}_x=\textbf{e}_x^{\prime}+\textbf{e}_y^{\prime}
\end{equation}
Similarly, $\textbf{e}_z$ consists of two error constituents – Pauli $Z$ errors and Pauli $Y$ errors, which can be denoted by vectors $\textbf{e}_z^{\prime}$ and $\textbf{e}_y^{\prime}$, respectively. Namely,
\begin{equation}
	\label{ez}
	\textbf{e}_z=\textbf{e}_z^{\prime}+\textbf{e}_y^{\prime}
\end{equation}
For instance, the symplectic representation of Pauli error $X_1 Y_2 Z_3$ is 
\begin{equation}
	\textbf{e}=(\textbf{e}_x\mid \textbf{e}_z )=(1\ 1\ 0\mid 0\ 1\ 1)
\end{equation}
where $\textbf{e}_x=(1\ 1\ 0)$ and $\textbf{e}_z=(0\ 1\ 1)$. According to Eq. (\ref{ex}) and (\ref{ez}) $\textbf{e}_x=(1\ 1\ 0)=\textbf{e}_x^{\prime}+\textbf{e}_y^{\prime}=\\
(1\ 0\ 0)+(0\ 1\ 0)$ and $\textbf{e}_z=(0\ 1\ 1)=\textbf{e}_z^{\prime}+\textbf{e}_y^{\prime}=(0\ 0\ 1)+(0\ 1\ 0)$, where $\textbf{e}_x^{\prime}=(1\ 0\ 0)$ represents the error $X_1 I_2 I_3$, $\textbf{e}_y^{\prime}=(0\ 1\ 0)$ represents the error $I_1 Y_2 I_3$ and $\textbf{e}_z^{\prime}=(0\ 0\ 1)$ represents the error $I_1 I_2 Z_3$.

In symplectic representation, the above analysis shows that Pauli $Y$ errors introduce the correlation between vectors $\textbf{e}_x$ and $\textbf{e}_z$, which should be taken into account in decoding. However, most decoding algorithms of QECCs decode $\textbf{e}_x$ and $\textbf{e}_z$ separately and omit the correlation introduced by Pauli $Y$ errors. The reason is that symplectic representation does not explicitly provide us with a vector to represent Pauli $Y$ errors, due to which omitting the correlation seems to be natural. The idea of decoupling representation comes from it — we try to propose a representation method in which all Pauli errors will be naturally taken into account.

The idea of decoupling representation is using a binary vector $\textbf{e}$ of size $3n$, which consists of three different parts of vectors $\textbf{e}_x^{\prime}$, $\textbf{e}_y^{\prime}$ and $\textbf{e}_z^{\prime}$, to represent Pauli errors $E$ acting on $n$ qubits. $\textbf{e}_x^{\prime}$, $\textbf{e}_y^{\prime}$ and $\textbf{e}_z^{\prime}$ represent Pauli $X$, $Y$ and $Z$ constituents of $E$, respectively, so that there is no correlation between the representation of these three Pauli errors constituent in mathematical. The decoupling representation of Pauli errors is defined as follow.

\begin{definition}[\textbf{Decoupling representation of Pauli operators}]
	\label{Decoupling representation of Pauli operators}
	The representation which represents Pauli operators by the following mapping is called decoupling representation.
	\begin{equation}
		\begin{aligned}
			&I\rightarrow(0,0,0),\ X\rightarrow(1,0,0),\\
			&Z\rightarrow(0,1,0),\ Y\rightarrow(0,0,1)
		\end{aligned}	
	\end{equation}
	For a Pauli error $E$ acting on $n$ qubits, by the above mapping, its decoupling representation is a binary vector $\textbf{e}$ with size of $3n$, namely,
	\begin{equation}
		\textbf{e}=(\textbf{e}_x^{\prime}\mid\textbf{e}_z^{\prime}\mid\textbf{e}_y^{\prime})
	\end{equation}
	where $\textbf{e}_x^{\prime}$, $\textbf{e}_z^{\prime}$ and $\textbf{e}_y^{\prime}$ are all binary vectors with size of $n$. For instance, the decoupling representation of $X_1 Y_2 Z_3$ is $\textbf{e}=(\textbf{e}_x^{\prime}\mid\textbf{e}_z^{\prime}\mid\textbf{e}_y^{\prime})=(1\ 0\ 0\mid0\ 0\ 1\mid0\ 1\ 0 )$.
\end{definition}

\begin{definition}[\textbf{Decoupled parity-check matrix}]
	\label{Decoupled parity-check matrix}
	Given an $[[n,k]]$ QLDPC code $C$ and the symplectic representation of its stabilizer generators $H=(H_x\mid H_z)$, the decoupled parity-check matrix of $C$ is
	\begin{equation}
		H_d=(H_z\mid H_x \mid(H_x\oplus H_z))
	\end{equation}
	whose dimension is $(n-k)\times 3n$ and where $\oplus$ denotes addition modulo 2. We use notation $H_d(j,i)$ to denote the element in the $j$th row and the $i$th column.
\end{definition}

\begin{lemma}[\textbf{The sum of elements of decoupled parity-check matrix}]
	\label{corollary 1}
	For the decoupled parity-check matrix $H_d=(H_z\mid H_x\mid(H_x\oplus H_z))$ of a $[[n,k]]$ QLDPC code, the number of $1$ among $i$th, $(i+n)$th and $(i+2n)$th column elements of each row is either $0$ or $2$, namely,
	\begin{equation}
		\begin{aligned}
			&H_d(j,i)+H_d(j,i+n)+H_d(j,i+2n)\\
			&=\left\{ 
			\begin{array}{lc}
				0, & if\ the\ jth\ stabilizer\ generator\ acts\\
				&trivially\ on\ the\ ith\ qubit \\
				2,& otherwise\\
			\end{array}
			\right.
		\end{aligned}
	\end{equation}
	for all $1\le j\le n-k$, $1\le i\le n$.
\end{lemma}

\begin{proof}
	1. If the $j$th stabilizer generator acts trivially on the $i$th qubit, then
	\begin{equation}
		H_d (j,i)=H_d (j,i+n)=H_d (j,i+2n)=0
	\end{equation}
	2. If the $j$th stabilizer generator acts nontrivially on the $i$th qubit, there are $3$ cases in total.
	
	A. If a Pauli $X$ error acts on the $i$th qubit, then
	\begin{equation}
		H_d (j,i+n)=H_d (j,i+2n)=1,\ H_d (j,i)=0
	\end{equation}
	
	B. If a Pauli $Y$ error acts on the $i$th qubit, then
	\begin{equation}
		H_d (j,i)=H_d (j,i+n)=1,\ H_d (j,i+2n)=0
	\end{equation}
	
	C. If a Pauli $Z$ error acts on the $i$th qubit, then
	\begin{equation}
		H_d (j,i)=H_d (j,i+2n)=1,\ H_d (j,i+n)=0
	\end{equation}
	No matter which case, we must have
	\begin{equation}
		H_d (j,i)+H_d (j,i+n)+H_d (j,i+2n)=2
	\end{equation}
	which completes the proof.
\end{proof}

\begin{lemma}[\textbf{The error syndrome expression in the decoupling representation}]
	Given the decoupling representation $\textbf{e}=(\textbf{e}_x^{\prime}\mid\textbf{e}_z^{\prime}\mid\textbf{e}_y^{\prime})$ of a Pauli error $E$ and a parity-check matrix $H_d=(H_z\mid H_x\mid(H_x\oplus H_z))$ of a QLDPC code, the error syndrome $\textbf{s}$ is
	\begin{equation}
		\textbf{s}=(H_d\cdot \textbf{e})\ mod\ 2
	\end{equation}
\end{lemma}

\begin{proof}
	
	\begin{equation}
		\begin{aligned}
			\textbf{s}&=(H_d\cdot \textbf{e})\ mod\ 2 \\
			&=((H_z\mid H_x\mid(H_x\oplus H_z ))\cdot(\textbf{e}_x^{\prime}\mid\textbf{e}_z^{\prime}\mid\textbf{e}_y^{\prime}))\ mod\ 2\\
			&=(H_z\cdot \textbf{e}_x^{\prime}+H_x\cdot \textbf{e}_z^{\prime}+(H_x\oplus H_z )\cdot \textbf{e}_y^{\prime})\ mod\ 2\\
			&=(H_z\cdot(\textbf{e}_x^{\prime}+\textbf{e}_y^{\prime})+H_x\cdot(\textbf{e}_z^{\prime}+\textbf{e}_y^{\prime}))\ mod\ 2\\
			&=(H_z\cdot \textbf{e}_x+H_x\cdot \textbf{e}_z )\ mod\ 2\\
			&=((H_z\mid H_x )\cdot(\textbf{e}_x\mid\textbf{e}_z ))\ mod\ 2
		\end{aligned}
	\end{equation}
	which is consistent with Eq. (\ref{syndrome}) and completes the proof.
\end{proof}

\subsection {Partially decoupled binary BP algorithm}
\label{3.2}
In this section, we will introduce PDBP algorithm. Different from SBP whose message passing is performed on the Tanner graph with $2n$ variable nodes corresponding to the symplectic parity-check matrix, PDBP perform message passing on the Tanner graph with $3n$ variable nodes corresponding to the decoupled parity-check matrix, while the rules of message update and hard decision have not been modified yet. As a contrast, the rules of message update and hard decision are modified to be fully compatible with decoupling representation in FDBP, which will be introduced in Sect. \ref{3.3}.

\textbf{Algorithm} \ref{Min-sum-based PDBP} shows min-sum-based PDBP algorithm. \textbf{Algorithm} \ref{sum-product-based PDBP} shows sum-product-based PDBP algorithm. In PDBP, before decoding, we should assign the error probability to each bit of the decoupling representation $\textbf{e}=(\textbf{e}_x^{\prime}\mid\textbf{e}_z^{\prime}\mid\textbf{e}_y^{\prime})=(e_1^{\prime},\cdots,e_n^{\prime}\mid e_{n+1}^{\prime},\cdots,e_{2n}^{\prime}\mid e_{2n+1}^{\prime},\\
\cdots,e_{3n}^{\prime})$ of a Pauli error $E$ acting on $n$ qubits. Since we only consider independent single-qubit Pauli noise channel whose error probability distribution is $(1-p_X-p_Z-p_Y,\ p_X,\ p_Y,\ p_Z )$, the bits of $\textbf{e}_x^{\prime}$, $\textbf{e}_z^{\prime}$ and $\textbf{e}_y^{\prime}$ are independently flipped with probability $p_X$, $p_Z$, and $p_Y$, respectively. Thus, the error rate vector on each bit of $\textbf{e}$ is $\textbf{\textit{p}}=(p_1,\cdots,p_n \mid p_{n+1},\cdots,p_{2n}\mid p_{2n+1},\cdots,p_{3n})$, where $p_1=\cdots=p_n=p_X$, $p_{n+1}=\cdots=p_{2n}=p_Z$, and $p_{2n+1}=\cdots=p_{3n}=p_Y$. The rules of message update and hard decision are the same as those of min-sum-based and sum-product-based BP in \cite{mackay1999good,fossorier1999reduced}, to be brief, we do not describe them in detail, and readers can see \cite{mackay1999good,fossorier1999reduced} for more details. 

\begin{breakablealgorithm}
	\caption{Min-sum-based PDBP}
	\label{Min-sum-based PDBP}
	\begin{algorithmic}[1]
		\textbf{Input:} decoupled parity-check matrix $H_d=(H_z |H_x |(H_x\oplus H_z ))$ with dimension of $m\times 3n$,\\
		error rate vector $\textbf{\textit{p}}=(p_1,\cdots,p_n \mid p_{n+1},\cdots,p_{2n}\mid p_{2n+1},\cdots,p_{3n})$ with size of $3n$, where $p_1=\cdots=p_n=p_X$, $p_{n+1}=\cdots=p_{2n}=p_Z$, and $p_{2n+1}=\cdots=p_{3n}=p_Y$,\\
		the maximum iteration number $iter_{max}$,\\
		syndrome vector $\textbf{s}=(s_1,\cdots,s_m )$ with size of $m$.
		
		\textbf{Output:} binary error vector $\hat{e}$ with size of $3n$.
		
		\textbf{Initialization:}
		
		\textbf{for} $i\leftarrow1$ to $3n$, let
		\begin{equation}
			\gamma_i = \ln \frac{1-p_i}{p_i}
		\end{equation}
		
		\textbf{for} each $v_i\in V$, $c_j\in C$ such that $H_d (j,i)=1$, let
		\begin{equation}
			m_{v_i\Rightarrow c_j}= \gamma_i 
		\end{equation}
		
		\textbf{Horizontal update:}
		
		\textbf{for} $j\leftarrow1$ to $m$, do
		\begin{equation}
			\begin{aligned}
				&m_{c_j\Rightarrow v_i}= (-1)^{s_j}\left(\prod_{v_{i^{\prime}}\in \mathcal{N}(c_j)\slash v_i}sign(m_{v_{i^{\prime}}\Rightarrow c_j})\right.\\
				&\left. \times \min_{v_{i^{\prime}}\in \mathcal{N}(c_j)\slash v_i}(\lvert m_{v_{i^{\prime}}\Rightarrow c_j} \rvert)\right)
			\end{aligned}
		\end{equation}
		
		\textbf{Vertical update:}
		
		\textbf{for} $i\leftarrow1$ to $3n$, do
		\begin{equation}
			m_{v_{i}\Rightarrow c_j}=\gamma_i + \sum_{c_{j^{\prime}}\in \mathcal{N}(v_i)/c_j }m_{c_{j^{\prime}}\Rightarrow v_i}
		\end{equation}
		
		\textbf{Hard decision:}
		
		\textbf{for} $i\leftarrow1$ to $3n$, do
		\begin{equation}
			\tilde{\gamma}_i=\gamma_i + \sum_{c_{j^{\prime}}\in \mathcal{N}(v_i)}m_{c_{j^{\prime}}\Rightarrow v_i}
		\end{equation}
		
		\textbf{if} $\tilde{\gamma}_i < 0$, let
		\begin{equation}
			\hat{e_i}=1
		\end{equation}
		
		\textbf{else}, let
		\begin{equation}
			\hat{e_i}=0
		\end{equation}
		
		if $H_d\cdot\hat{e}=\textbf{s}$, return “\textbf{converge}” and halt.
		
		else if the maximum iteration number $iter_{max}$ is reached, return “\textbf{fail}” and halt.
		
		else repeat from \textbf{Horizontal update}.
	\end{algorithmic}
\end{breakablealgorithm}

\begin{breakablealgorithm}
	\caption{Sum-product-based PDBP}
	\label{sum-product-based PDBP}
	\begin{algorithmic}
		\textbf{Input:} decoupled parity-check matrix $H_d=(H_z |H_x |(H_x\oplus H_z ))$ with dimension of $m\times 3n$,\\
		error rate vector $\textbf{\textit{p}}=(p_1,\cdots,p_n \mid p_{n+1},\cdots,p_{2n}\mid p_{2n+1},\cdots,p_{3n})$ with size of $3n$, where $p_1=\cdots=p_n=p_X$, $p_{n+1}=\cdots=p_{2n}=p_Z$, and $p_{2n+1}=\cdots=p_{3n}=p_Y$,\\
		the maximum iteration number $iter_{max}$,\\
		syndrome vector $\textbf{s}=(s_1,\cdots,s_m )$ with size of $m$.
		
		\textbf{Output:} binary error vector $\hat{e}$ with size of $3n$.
		
		\textbf{Initialization:}
		
		\textbf{for} $i\leftarrow1$ to $3n$, let
		\begin{equation}
			\gamma_i = \ln \frac{1-p_i}{p_i}
		\end{equation}
		
		\textbf{for} each $v_i\in V$, $c_j\in C$ such that $H_d (j,i)=1$, let
		\begin{equation}
			m_{v_i\Rightarrow c_j}= \gamma_i 
		\end{equation}
		
		\textbf{Horizontal update:}
		
		\textbf{for} $j\leftarrow1$ to $m$, do
		\begin{equation}
			w = \ln \frac{1+\prod_{v_{i^{\prime}}\in \mathcal{N}(c_j)\slash v_i}\tanh \left(\frac{m_{v_{i^{\prime}}\Rightarrow c_j}}{2}\right)}{1-\prod_{v_{i^{\prime}}\in \mathcal{N}(c_j)\slash v_i}\tanh \left(\frac{m_{v_{i^{\prime}}\Rightarrow c_j}}{2}\right)}
		\end{equation}
		
		\begin{equation}
			m_{c_j\Rightarrow v_i}= (-1)^{s_j}w
		\end{equation}
		
		\textbf{Vertical update:}
		
		\textbf{for} $i\leftarrow1$ to $3n$, do
		\begin{equation}
			m_{v_{i}\Rightarrow c_j}=\gamma_i + \sum_{c_{j^{\prime}}\in \mathcal{N}(v_i)/c_j }m_{c_{j^{\prime}}\Rightarrow v_i}
		\end{equation}
		
		\textbf{Hard decision:}
		
		\textbf{for} $i\leftarrow1$ to $3n$, do
		\begin{equation}
			\tilde{\gamma}_i=\gamma_i + \sum_{c_{j^{\prime}}\in \mathcal{N}(v_i)}m_{c_{j^{\prime}}\Rightarrow v_i}
		\end{equation}
		
		\textbf{if} $\tilde{\gamma}_i < 0$, let
		\begin{equation}
			\hat{e_i}=1
		\end{equation}
		
		\textbf{else}, let
		\begin{equation}
			\hat{e_i}=0
		\end{equation}
		
		if $H_d\cdot\hat{e}=\textbf{s}$, return “\textbf{converge}” and halt.
		
		else if the maximum iteration number $iter_{max}$ is reached, return “\textbf{fail}” and halt.
		
		else repeat from \textbf{Horizontal update}.
	\end{algorithmic}
\end{breakablealgorithm}

The simulation results of min-sum-based and sum-product-based PDBP on XZZX surface code and planar surface code are shown in Sect. \ref{4}.

\subsection {Fully decoupled binary BP algorithm}
\label{3.3}
Careful readers must have noticed that, according to Corollary \ref{corollary 1}, in decoupling representation $\textbf{e}=(\textbf{e}_x^{\prime}\mid\textbf{e}_z^{\prime}\mid\textbf{e}_y^{\prime})$ of a Pauli error acting on $n$ qubits, the total number of $1$s of the $i$th $(i\leq n)$, $(i+n)$th and $(i+2n)$th elements of $\textbf{e}$ is no more than one, namely,
\begin{equation}
	\label{restraint condition}
	e_i+e_{i+n}+e_{i+2n}\leq1
\end{equation}
which should be taken into account in decoding. Based on PDBP algorithm, we take this into account and obtain FDBP algorithm. In FDBP, the restraint condition Eq. (\ref{restraint condition}) is reflected in message update and hard decision processes. Before introducing FDBP algorithm, we need an extra notation.

\begin{itemize}
	\item
	\textbf{For the $i$th column of a decoupled parity-check matrix, we use $i_s$ to denote the column orders corresponding to the same qubit pointed by $i$, namely, $i_s\in \{(i+n)\ mod
		\\ 3n,\ (i+2n)\ mod\ 3n\}$.}
\end{itemize}

\textbf{Algorithm} \ref{Min-sum-based FDBP} and \textbf{Algorithm} \ref{Sum-product-based FDBP} show the min-sum-based and sum-product-based FDBP, respectively. Their mathematical derivations are given in Appendix \ref{Appendix B} and Appendix \ref{Appendix A}, respectively. The update rule of the check-to-variable message $m_{c_j\Rightarrow v_i}$ of FDBP is quite different from that of PDBP. The difference comes from the restraint condition that if $\hat{e}_i=1$, $\hat{e}_{i_s}=0$. Readers can see Appendix \ref{Appendix A} and Appendix \ref{Appendix B} for more details of its mathematical derivation.

\begin{breakablealgorithm}
	\caption{Min-sum-based FDBP}
	\label{Min-sum-based FDBP}
	\begin{algorithmic}
		\textbf{Input:} decoupled parity-check matrix $H_d=(H_z |H_x |(H_x\oplus H_z ))$ with dimension of $m\times 3n$,\\
		error rate vector $\textbf{\textit{p}}=(p_1,\cdots,p_n \mid p_{n+1},\cdots,p_{2n}\mid p_{2n+1},\cdots,p_{3n})$ with size of $3n$, where $p_1=\cdots=p_n=p_X$, $p_{n+1}=\cdots=p_{2n}=p_Z$, and $p_{2n+1}=\cdots=p_{3n}=p_Y$,\\
		the maximum iteration number $iter_{max}$,\\
		syndrome vector $\textbf{s}=(s_1,\cdots,s_m )$ with size of $m$.
		
		\textbf{Output:} binary error vector $\hat{e}$ with size of $3n$.
		
		\textbf{Initialization:}
		
		\textbf{for} $i\leftarrow1$ to $3n$, let
		\begin{equation}
			\gamma_i = \ln \frac{1-p_i}{p_i}
		\end{equation}
		
		\textbf{for} each $v_i\in V$, $c_j\in C$ such that $H_d (j,i)=1$, let
		\begin{equation}
			m_{v_i\Rightarrow c_j}= \gamma_i 
		\end{equation}
		
		\textbf{Horizontal update:}
		
		\textbf{for} $j\leftarrow1$ to $m$, do
		
		\textbf{if} $s_j=0$
		\begin{equation}
			\begin{aligned}
				&m_{c_j\Rightarrow v_i}= \left(\prod_{v_{i^{\prime}}\in \mathcal{N}(c_j)\slash v_i}sign(m_{v_{i^{\prime}}\Rightarrow c_j})\right.\\
				&\times \left.\min_{v_{i^{\prime}}\in \mathcal{N}(c_j)\slash v_i}(\lvert m_{v_{i^{\prime}}\Rightarrow c_j} \rvert)\right)\\
				&+\ln \frac{1-\prod_{v_{i^{\prime}}\in \mathcal{N}(c_j)\slash v_i}\tanh \left(\frac{m_{v_{i^{\prime}}\Rightarrow c_j}}{2}\right)}{1-\prod_{v_{i^{\prime}}\in \mathcal{N}(c_j)\slash \{v_i,v_{i_s}\}} \tanh \left(\frac{m_{v_{i^{\prime}}\Rightarrow c_j}}{2}\right)}
			\end{aligned}
		\end{equation}
		
		\textbf{if} $s_j=1$
		\begin{equation}
			\begin{aligned}
				&m_{c_j\Rightarrow v_i}= \left(-1\prod_{v_{i^{\prime}}\in \mathcal{N}(c_j)\slash \{v_i,v_{i_s}\}}sign(m_{v_{i^{\prime}}\Rightarrow c_j})\right.\\
				&\left.\times \min_{v_{i^{\prime}}\in \mathcal{N}(c_j)\slash \{v_i,v_{i_s}\}}(\lvert m_{v_{i^{\prime}}\Rightarrow c_j} \rvert)\right)\\
				& +\ln \frac{1-\prod_{v_{i^{\prime}}\in \mathcal{N}(c_j)\slash v_i}\tanh \left(\frac{m_{v_{i^{\prime}}\Rightarrow c_j}}{2}\right)}{1-\prod_{v_{i^{\prime}}\in \mathcal{N}(c_j)\slash \{v_i,v_{i_s}\}} \tanh \left(\frac{m_{v_{i^{\prime}}\Rightarrow c_j}}{2}\right)}
			\end{aligned}
		\end{equation}
		
		\textbf{Vertical update:}
		
		\textbf{for} $i\leftarrow1$ to $3n$, do
		\begin{equation}
			m_{v_{i}\Rightarrow c_j}=\gamma_i + \sum_{c_{j^{\prime}}\in \mathcal{N}(v_i)/c_j }\left[ m_{c_{j^{\prime}}\Rightarrow v_i}-\ln (1-p_{i_s})\right]
		\end{equation}
		
		\textbf{Hard decision:}
		
		\textbf{for} $i\leftarrow1$ to $3n$, do
		\begin{equation}
			\tilde{\gamma}_i=\gamma_i + \sum_{c_{j^{\prime}}\in \mathcal{N}(v_i)}\left[ m_{c_{j^{\prime}}\Rightarrow v_i}-\ln (1-p_{i_s})\right]
		\end{equation}
		
		\textbf{for} $i\leftarrow1$ to $n$, do
		
		\textbf{if} $\tilde{\gamma}_{i+kn} \ge 0$ for all $k\in\{0,1,2\}$, let
		\begin{equation}
			\hat{e}_{i+kn}=0,\ for\ all\ k\in\{0,1,2\}
		\end{equation}
		
		\textbf{else}, let
		\begin{equation}
			m=\mathop{\arg\min}\limits_{k\in\{0,1,2\}}\ \tilde{\gamma}_{i+kn}
		\end{equation}
		
		\begin{equation}
			\hat{e}_{i+mn}=1
		\end{equation}
		
		and
		\begin{equation}
			\hat{e}_{i+k^{\prime}n}=0
		\end{equation}
		for all $k^{\prime}\in\{0,1,2\}/\{m\}$.
		
		if $H_d\cdot\hat{e}=\textbf{s}$, return “\textbf{converge}” and halt.
		
		else if the maximum iteration number $iter_{max}$ is reached, return “\textbf{fail}” and halt.
		
		else repeat from \textbf{Horizontal update}.
	\end{algorithmic}
\end{breakablealgorithm}

\begin{breakablealgorithm}
	\caption{Sum-product-based FDBP}
	\label{Sum-product-based FDBP}
	\begin{algorithmic}
		\textbf{Input:} decoupled parity-check matrix $H_d=(H_z |H_x |(H_x\oplus H_z ))$ with dimension of $m\times 3n$,\\
		error rate vector $\textbf{\textit{p}}=(p_1,\cdots,p_n \mid p_{n+1},\cdots,p_{2n}\mid p_{2n+1},\cdots,p_{3n})$ with size of $3n$, where $p_1=\cdots=p_n=p_X$, $p_{n+1}=\cdots=p_{2n}=p_Z$, and $p_{2n+1}=\cdots=p_{3n}=p_Y$,\\
		the maximum iteration number $iter_{max}$,\\
		syndrome vector $\textbf{s}=(s_1,\cdots,s_m )$ with size of $m$.
		
		\textbf{Output:} binary error vector $\hat{e}$ with size of $3n$.
		
		\textbf{Initialization:}
		
		\textbf{for} $i\leftarrow1$ to $3n$, let
		\begin{equation}
			\gamma_i = \ln \frac{1-p_i}{p_i}
		\end{equation}
		
		\textbf{for} each $v_i\in V$, $c_j\in C$ such that $H_d (j,i)=1$, let
		\begin{equation}
			m_{v_i\Rightarrow c_j}= \gamma_i 
		\end{equation}
		
		\textbf{Horizontal update:}
		
		\textbf{for} $j\leftarrow1$ to $m$, do
		
		\textbf{if} $s_j=0$
		\begin{equation}
			m_{c_j\Rightarrow v_i}= \ln \frac{1+\prod_{v_{i^{\prime}}\in \mathcal{N}(c_j)\slash v_i}\tanh \left(\frac{m_{v_{i^{\prime}}\Rightarrow c_j}}{2}\right)}{1-\prod_{v_{i^{\prime}}\in \mathcal{N}(c_j)\slash \{v_i,v_{i_s}\}} \tanh \left(\frac{m_{v_{i^{\prime}}\Rightarrow c_j}}{2}\right)}
		\end{equation}
		
		\textbf{if} $s_j=1$
		\begin{equation}
			m_{c_j\Rightarrow v_i}= \ln \frac{1-\prod_{v_{i^{\prime}}\in \mathcal{N}(c_j)\slash v_i}\tanh \left(\frac{m_{v_{i^{\prime}}\Rightarrow c_j}}{2}\right)}{1+\prod_{v_{i^{\prime}}\in \mathcal{N}(c_j)\slash \{v_i,v_{i_s}\}} \tanh \left(\frac{m_{v_{i^{\prime}}\Rightarrow c_j}}{2}\right)}
		\end{equation}
		
		\textbf{Vertical update:}
		
		\textbf{for} $i\leftarrow1$ to $3n$, do
		\begin{equation}
			m_{v_{i}\Rightarrow c_j}=\gamma_i + \sum_{c_{j^{\prime}}\in \mathcal{N}(v_i)/c_j }\left[ m_{c_{j^{\prime}}\Rightarrow v_i}-\ln (1-p_{i_s})\right]
		\end{equation}
		
		\textbf{Hard decision:}
		
		\textbf{for} $i\leftarrow1$ to $3n$, do
		\begin{equation}
			\tilde{\gamma}_i=\gamma_i + \sum_{c_{j^{\prime}}\in \mathcal{N}(v_i)}\left[ m_{c_{j^{\prime}}\Rightarrow v_i}-\ln (1-p_{i_s})\right]
		\end{equation}
		
		\textbf{for} $i\leftarrow1$ to $n$, do
		
		\textbf{if} $\tilde{\gamma}_{i+kn} \ge 0$ for all $k\in\{0,1,2\}$, let
		\begin{equation}
			\hat{e}_{i+kn}=0,\ for\ all\ k\in\{0,1,2\}
		\end{equation}
		
		\textbf{else}, let
		\begin{equation}
			m=\mathop{\arg\min}\limits_{k\in\{0,1,2\}}\ \tilde{\gamma}_{i+kn}
		\end{equation}
		
		\begin{equation}
			\hat{e}_{i+mn}=1
		\end{equation}
		
		and
		\begin{equation}
			\hat{e}_{i+k^{\prime}n}=0
		\end{equation}
		for all $k^{\prime}\in\{0,1,2\}/\{m\}$.
		
		if $H_d\cdot\hat{e}=\textbf{s}$, return “\textbf{converge}” and halt.
		
		else if the maximum iteration number $iter_{max}$ is reached, return “\textbf{fail}” and halt.
		
		else repeat from \textbf{Horizontal update}.
	\end{algorithmic}
\end{breakablealgorithm}
The simulation results of min-sum-based and sum-product-based FDBP on XZZX surface code and planar surface code are shown in Sect. \ref{4}.

\section {Simulation results}
\label{4}
We perform comparative simulations between min-sum-based and sum-product-based SBP, PDBP and FDBP on XZZX surface code and planar surface code under different noise channels, including pure $X$, $Z$, $Y$ and depolarizing noise channels. All the simulations in this paper are combined with order$-0$ OSD and under the assumption that there is no measurement error in syndrome detection. We use open-source Python library \cite{Roffe_LDPC_Python_tools_2022} to implement SBP and PDBP to complete our simulations.
\subsection{Comparison between min-sum-based SBP, PDBP and FDBP}
\label{4.1}

Fig. \ref{fig:ms-based Comparison between SBP, PDBP and FDBP on XZZX surface code pure noise} shows the logical error rate (LER) as a function of physical qubit error rate of min-sum-based SBP, PDBP and FDBP on XZZX surface code with different lattice size $L$ under pure $X$, $Z$ and $Y$ noise. Fig. \ref{fig:ms-based Comparison between SBP, PDBP and FDBP on XZZX surface code pure noise} (a)$\sim$(f) show that under pure $X$ and $Z$ noise, the curves of min-sum-based SBP, PDBP and FDBP almost coincide perfectly, because Tanner graphs of SBP, PDBP and FDBP in this situation can be reduced to the same form by removing the variable nodes whose corresponding errors won’t happen in the noise model. Under pure $X$ and $Z$ noise, the code-capacity thresholds of min-sum-based SBP, PDBP and FDBP on XZZX surface code are all about $50\%$. However, for pure $Y$ noise, as shown in Fig. \ref{fig:ms-based Comparison between SBP, PDBP and FDBP on XZZX surface code pure noise} (g)$\sim$(i), PDBP and FDBP outperform SBP, while the performances of PDBP and FDBP are almost the same. Under pure $Y$ noise, the code-capacity thresholds of min-sum-based PDBP and FDBP on XZZX surface code are both about $40\%$, which are much higher than that of SBP, which is only about $9\%$. Similarly, Fig. \ref{fig:ms-based Comparison between SBP, PDBP and FDBP on planar surface code pure noise} shows the LER as a function of physical qubit error rate of min-sum-based SBP, PDBP and FDBP on planar surface code with different lattice size $L$ under pure $X$, $Z$ and $Y$ noise. Fig. \ref{fig:ms-based Comparison between SBP, PDBP and FDBP on planar surface code pure noise} (a)$\sim$(c) and (d)$\sim$(f) show the code-capacity thresholds of min-sum-based SBP, PDBP and FDBP on planar surface code under pure $X$ and $Z$ noise, respectively, which are all about $9\%$. While for pure $Y$ noise, as shown in Fig. \ref{fig:ms-based Comparison between SBP, PDBP and FDBP on planar surface code pure noise} (g)$\sim$(i), the code-capacity thresholds of min-sum-based PDBP and FDBP on planar surface code are both about $25\%$, which are also much higher than that of SBP, $9\%$. These results support that in the decoupling representation of Pauli errors decoders tend to have more balanced capability to decode Pauli $X$, $Y$ and $Z$ errors.

\begin{figure*}[htbp]
	\centering
	
	\begin{minipage}{1\linewidth}
		\centering
		\includegraphics[width=1\linewidth]{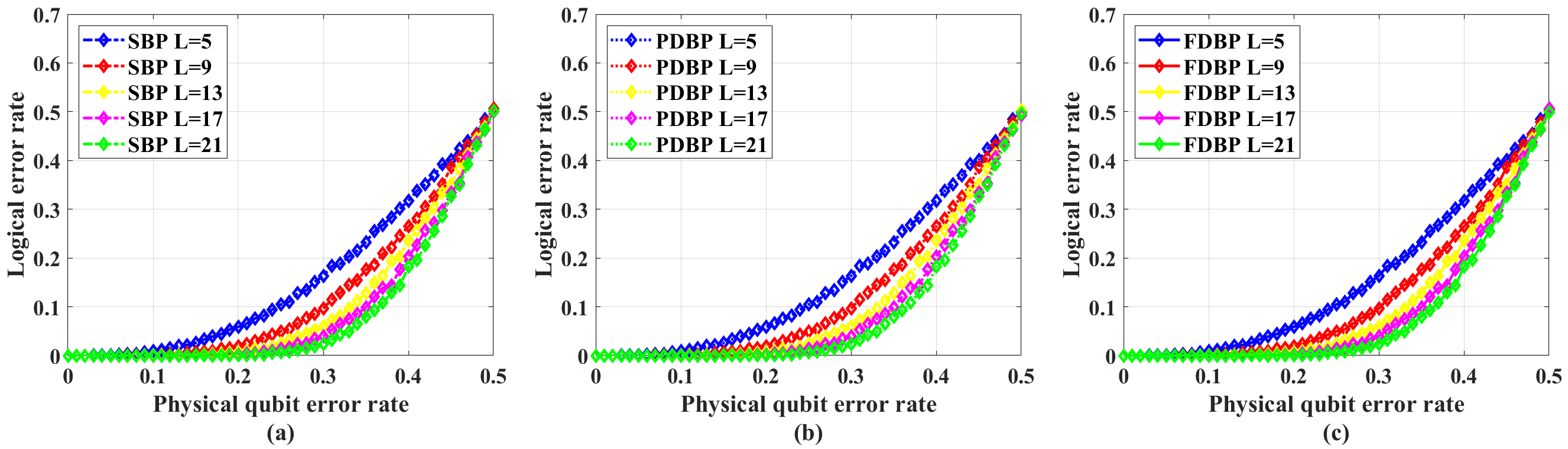}
	\end{minipage}
	\begin{minipage}{1\linewidth}
		\centering
		\includegraphics[width=1\linewidth]{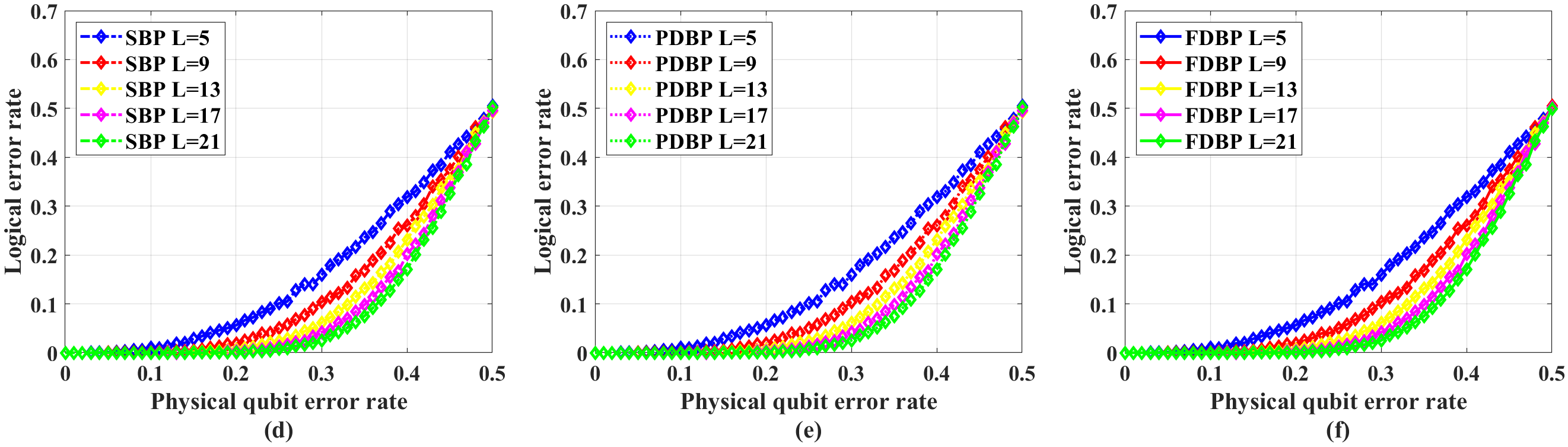}
	\end{minipage}
	\begin{minipage}{1\linewidth}
		\centering
		\includegraphics[width=1\linewidth]{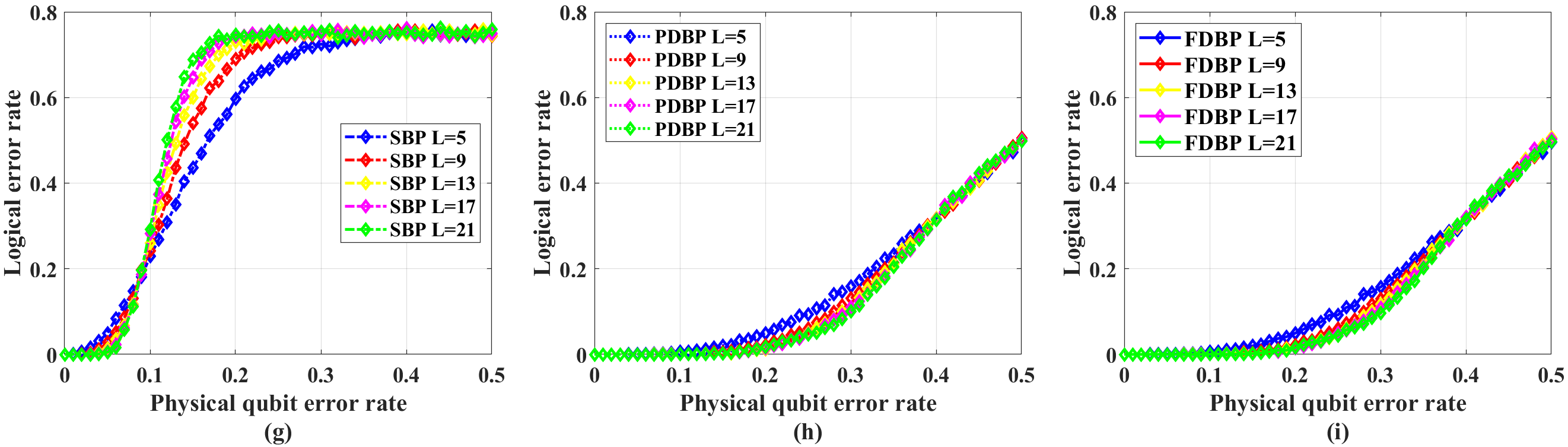}
	\end{minipage}
	\caption{\textbf{Logical error rate as a function of physical qubit error rate and code-capacity thresholds of min-sum-based SBP, PDBP and FDBP on XZZX surface code of different lattice size $L$ under pure $X$, $Z$ and $Y$ noise.}(a) Min-sum-based SBP on XZZX surface codes under pure $X$ noise. (b) Min-sum-based PDBP on XZZX surface codes under pure $X$ noise. (c) Min-sum-based FDBP on XZZX surface codes under pure $X$ noise. (d) Min-sum-based SBP on XZZX surface codes under pure $Z$ noise. (e) Min-sum-based PDBP on XZZX surface codes under pure $Z$ noise. (f) Min-sum-based FDBP on XZZX surface codes under pure $Z$ noise. (g) Min-sum-based SBP on XZZX surface codes under pure $Y$ noise. (h) Min-sum-based PDBP XZZX surface codes under pure $Y$ noise. (i) Min-sum-based FDBP on XZZX surface codes under pure $Y$ noise.}
	\label{fig:ms-based Comparison between SBP, PDBP and FDBP on XZZX surface code pure noise}
\end{figure*}

\begin{figure*}[htbp]
	\centering
	
	\begin{minipage}{1\linewidth}
		\centering
		\includegraphics[width=1\linewidth]{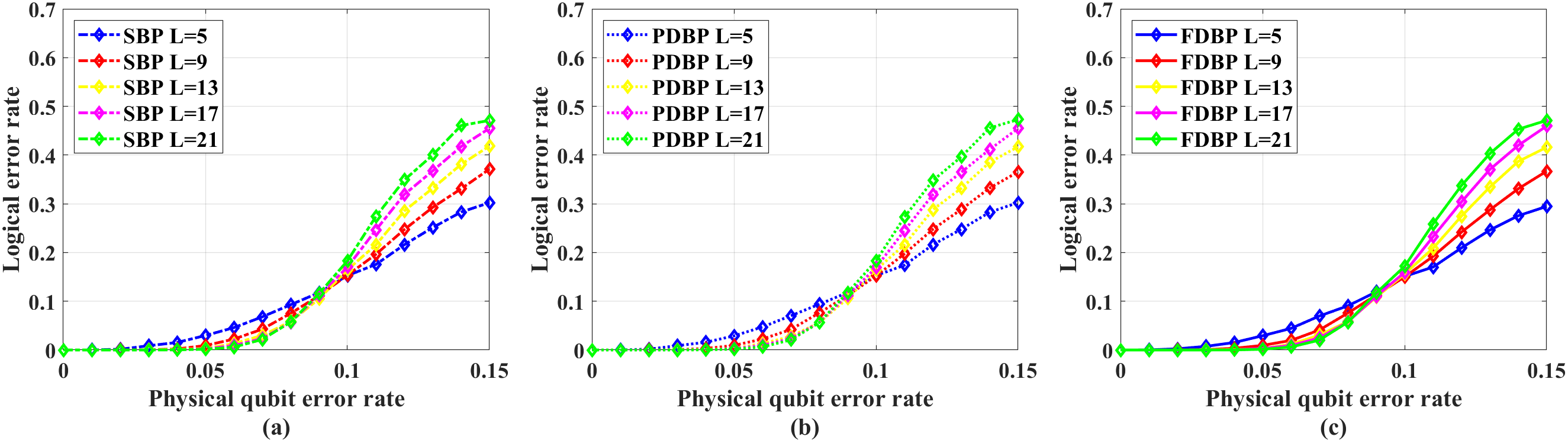}
	\end{minipage}
	\begin{minipage}{1\linewidth}
		\centering
		\includegraphics[width=1\linewidth]{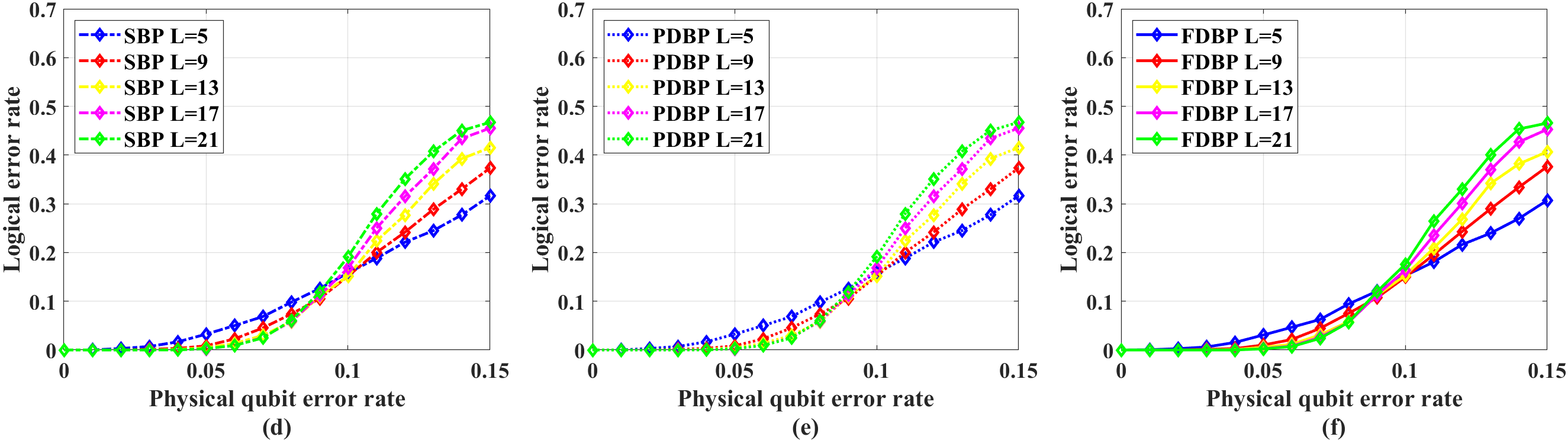}
	\end{minipage}
	\begin{minipage}{1\linewidth}
		\centering
		\includegraphics[width=1\linewidth]{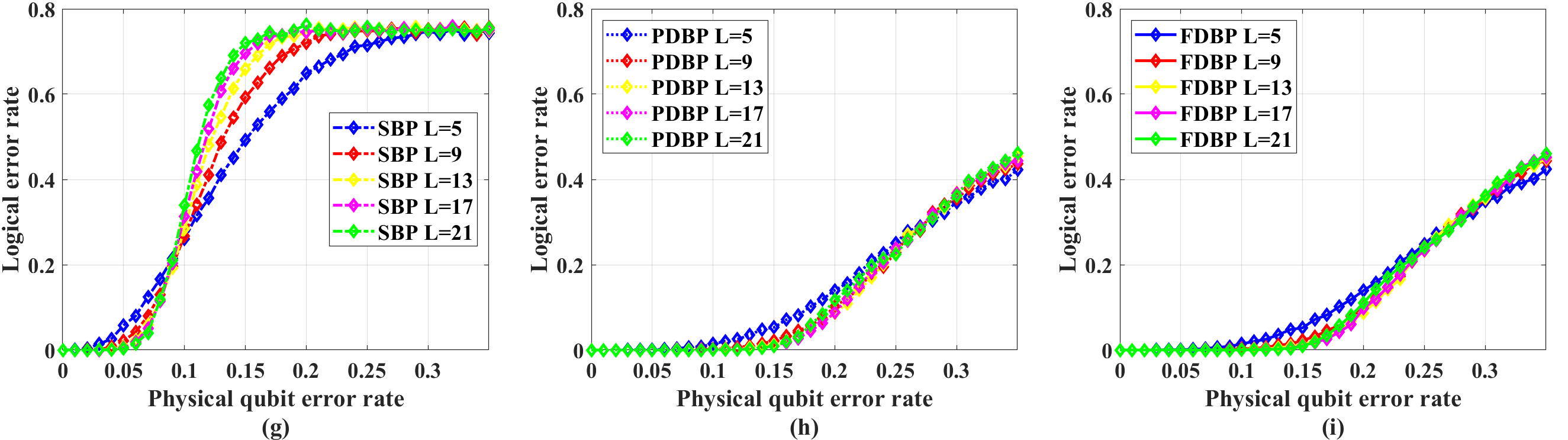}
	\end{minipage}
	\caption{\textbf{Logical error rate as a function of physical qubit error rate and code-capacity thresholds of min-sum-based SBP, PDBP and FDBP on planar surface code of different lattice size $L$ under pure $X$, $Z$ and $Y$ noise.}(a) Min-sum-based SBP on planar surface codes under pure $X$ noise. (b) Min-sum-based PDBP on planar surface codes under pure $X$ noise. (c) Min-sum-based FDBP on planar surface codes under pure $X$ noise. (d) Min-sum-based SBP on planar surface codes under pure $Z$ noise. (e) Min-sum-based PDBP on planar surface codes under pure $Z$ noise. (f) Min-sum-based FDBP on planar surface codes under pure $Z$ noise. (g) Min-sum-based SBP on planar surface codes under pure $Y$ noise. (h) Min-sum-based PDBP planar surface codes under pure $Y$ noise. (i) Min-sum-based FDBP on planar surface codes under pure $Y$ noise.}
	\label{fig:ms-based Comparison between SBP, PDBP and FDBP on planar surface code pure noise}
\end{figure*}

%

Due to the more balanced capability to decode Pauli $X$, $Y$ and $Z$ errors, PDBP and FDBP outperform SBP under depolarizing noise. As shown in Fig. \ref{fig:ms-based Comparison between SBP, PDBP and FDBP on XZZX surface code_jpg} (a)$\sim$(c), the code-capacity thresholds of min-sum-based SBP, PDBP and FDBP on XZZX surface code under depolarizing noise are about $13.7\%$, $13.8\%$ and $16\%$, respectively. And Fig. \ref{fig:ms-based Comparison between SBP, PDBP and FDBP on planar surface code_jpg} (a)$\sim$(c) show the code-capacity thresholds of min-sum-based SBP, PDBP and FDBP on planar surface code under depolarizing noise are about $14\%$, $13.5\%$ and $15.6\%$, respectively. Min-sum-based FDBP has the highest code-capacity threshold. In order to further compare their difference, we select three curves corresponding to different lattice size $L=5$, $13$ and $21$ from Fig. \ref{fig:ms-based Comparison between SBP, PDBP and FDBP on XZZX surface code_jpg} (a)$\sim$(c) respectively and obtain Fig. \ref{fig:ms-based Comparison between SBP, PDBP and FDBP on XZZX surface code_jpg} (d)$\sim$(f). Doing similarly for Fig. \ref{fig:ms-based Comparison between SBP, PDBP and FDBP on planar surface code_jpg} (a)$\sim$(c) and we obtain Fig. \ref{fig:ms-based Comparison between SBP, PDBP and FDBP on planar surface code_jpg} (d)$\sim$(f). We can see that, the LER of min-sum-based FDBP is the lowest. Moreover, the larger the lattice size is, the larger the difference between min-sum-based FDBP and another two algorithms is. Moreover, as shown in Fig. \ref{fig:ms-based Comparison between SBP, PDBP and FDBP on XZZX surface code_jpg} (g)$\sim$(i) and Fig. \ref{fig:ms-based Comparison between SBP, PDBP and FDBP on planar surface code_jpg} (g)$\sim$(i), the average number of iterations of min-sum-based FDBP is the least.

\begin{figure*}[htbp]
	\centering
	\begin{minipage}{1\linewidth}
		\centering
		\includegraphics[width=1\linewidth]{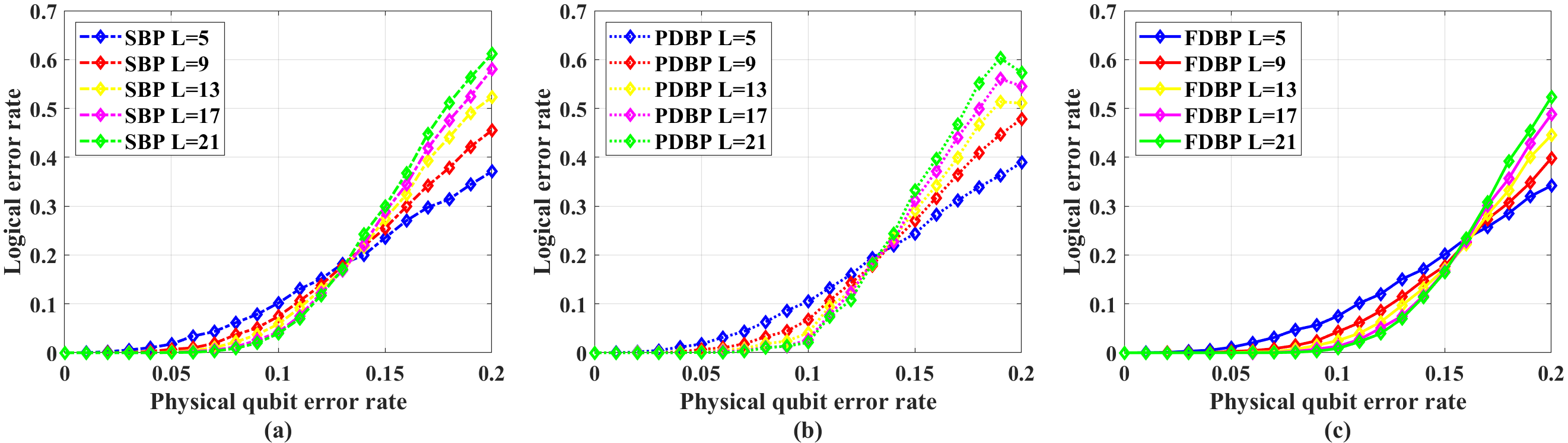}
	\end{minipage}
	\begin{minipage}{1\linewidth}
		\centering
		\includegraphics[width=1\linewidth]{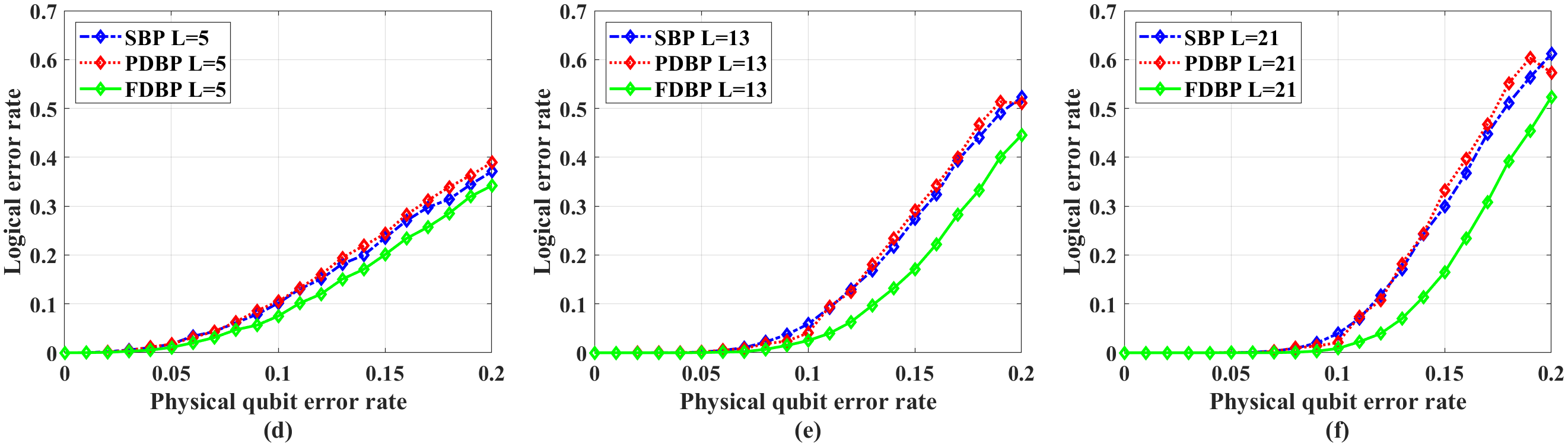}
	\end{minipage}
	\begin{minipage}{1\linewidth}
		\centering
		\includegraphics[width=1\linewidth]{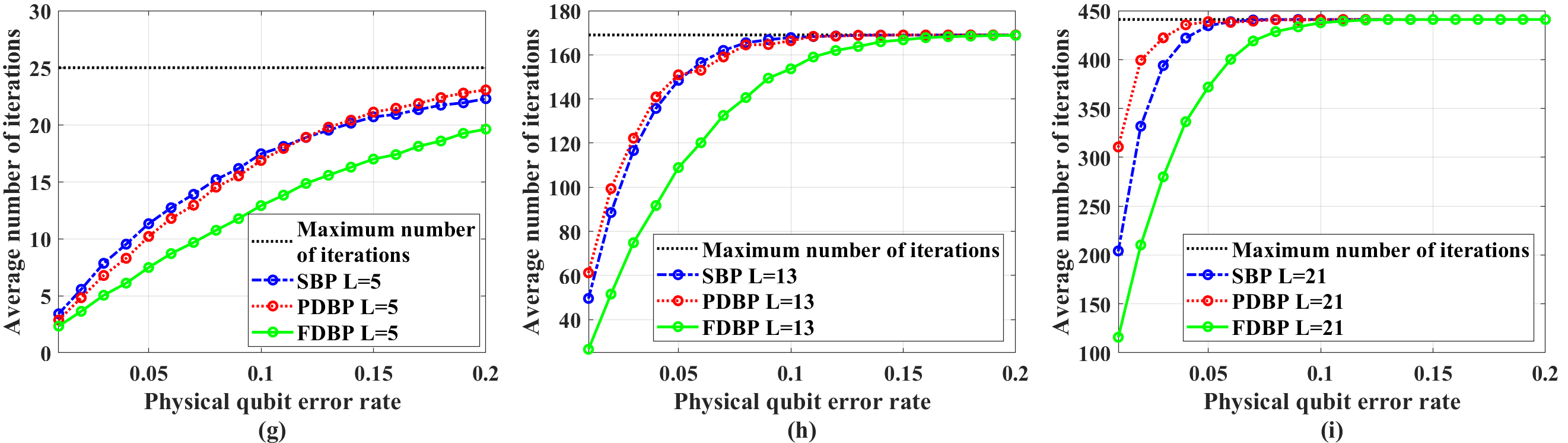}
	\end{minipage}
	\caption{Error-correcting performances of min-sum-based (a) SBP, (b) PDBP and (c) FDBP on XZZX surface code with different lattice size $L$ under depolarizing noise. We observe the code-capacity thresholds of min-sum-based SBP, PDBP and FDBP on XZZX surface code under depolarizing noise are $13.7\%$, $13.8\%$ and $16\%$, respectively. Comparison of LER of min-sum-based SBP, PDBP and FDBP on XZZX surface code with Lattice size (d) $L=5$, (e) $L=13$ and (e) $L=21$. The LER of FDBP is lowest and the larger the lattice size is, the larger the difference in logical error rate between min-sum-based FDBP and another two algorithms is. The average number of iterations of min-sum-based SBP, PDBP and FDBP on XZZX surface code with Lattice size (g) $L=5$, (h) $L=13$ and (i) $L=21$. The maximum number of iterations is the code length $N=L^2$. }
	\label{fig:ms-based Comparison between SBP, PDBP and FDBP on XZZX surface code_jpg}
\end{figure*}

\begin{figure*}[htbp]
	\centering
	\begin{minipage}{1\linewidth}
		\centering
		\includegraphics[width=1\linewidth]{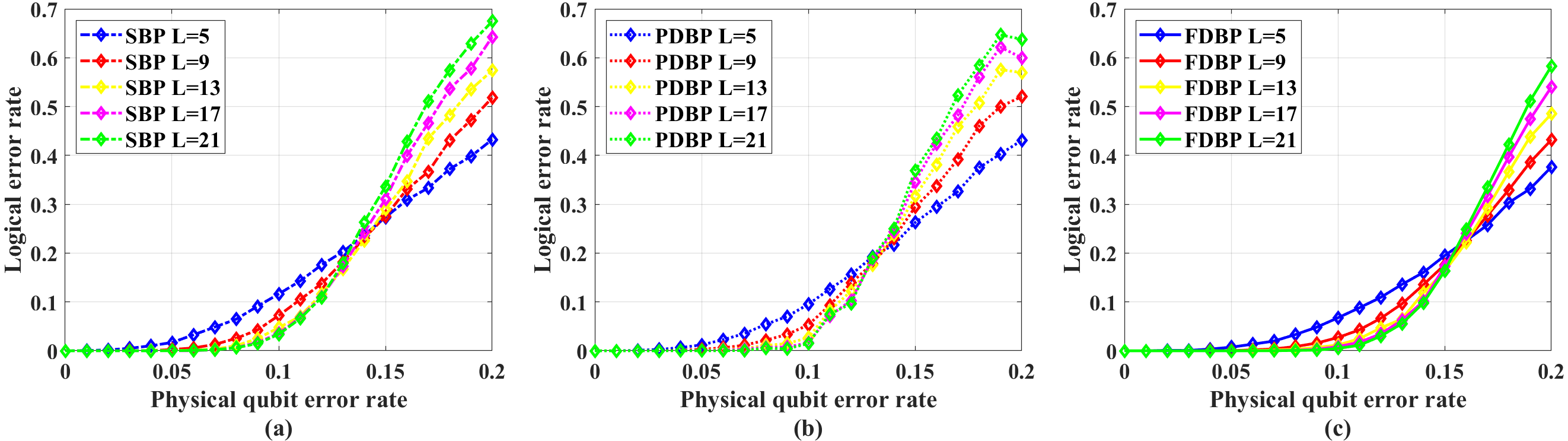}
	\end{minipage}
	\begin{minipage}{1\linewidth}
		\centering
		\includegraphics[width=1\linewidth]{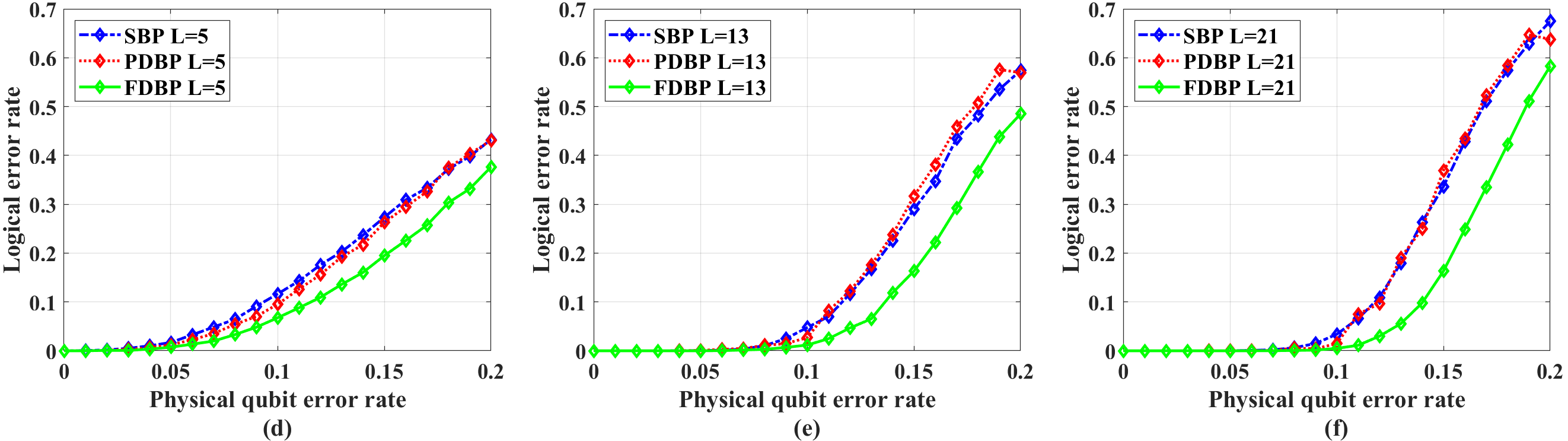}
	\end{minipage}
	\begin{minipage}{1\linewidth}
		\centering
		\includegraphics[width=1\linewidth]{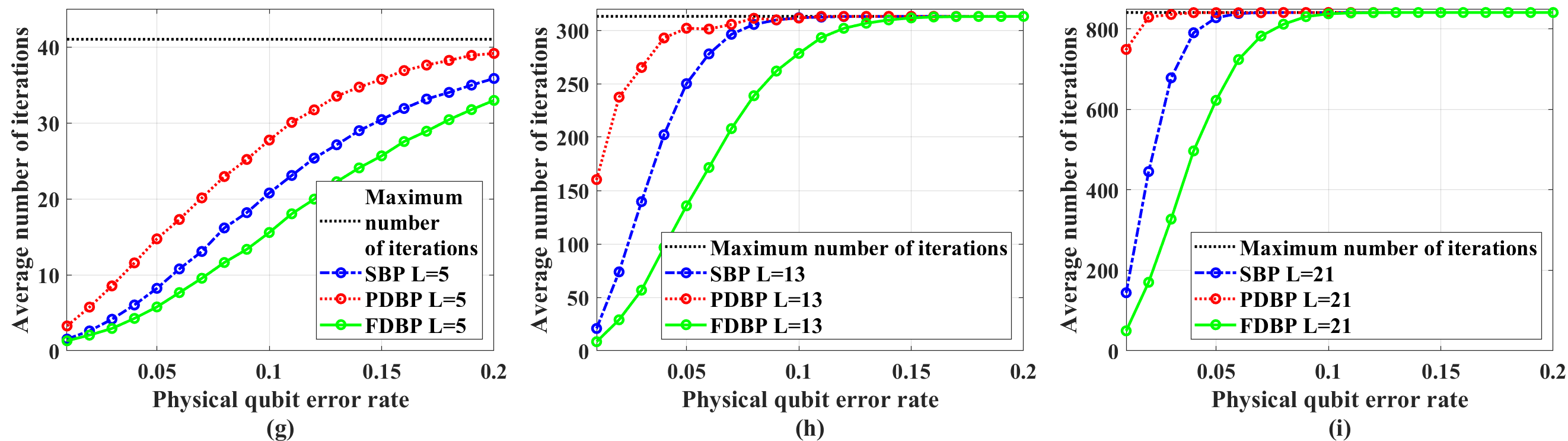}
	\end{minipage}
	\caption{Error-correcting performances of min-sum-based (a) SBP, (b) PDBP and (c) FDBP on planar surface code with different lattice size $L$ under depolarizing noise. We observe the code-capacity thresholds of min-sum-based SBP, PDBP and FDBP on planar surface code under depolarizing noise are $14\%$, $13.5\%$ and $15.6\%$, respectively. Comparison of LER of min-sum-based SBP, PDBP and FDBP on planar surface code with Lattice size (d) $L=5$, (e) $L=13$ and (e) $L=21$. The LER of FDBP is lowest and the larger the lattice size is, the larger the difference in logical error rate between min-sum-based FDBP and another two algorithms is. The average number of iterations of min-sum-based SBP, PDBP and FDBP on planar surface code with Lattice size (g) $L=5$, (h) $L=13$ and (i) $L=21$. The maximum number of iterations is the code length $N=2L^2-2L+1$. }
	\label{fig:ms-based Comparison between SBP, PDBP and FDBP on planar surface code_jpg}
\end{figure*}

Readers may wonder since PDBP and FDBP shows almost the same error-correcting performance under pure X, Z and Y noise, namely, they have almost the same capability to decode Pauli $X$, $Y$ and $Z$ errors when only one of these three kinds of Pauli errors might happen, why FDBP outperforms PDBP under depolarizing noise where Pauli $X$, $Y$ and $Z$ errors may happen with the same probability? The reason is that, under depolarizing noise, to be fully compatible with the decoupling representation of Pauli errors, the message update rules and hard decision of FDBP take the restraint condition Eq. (\ref{restraint condition}) into account, while PDBP does not. Thus, the performance of FDBP is better than that of PDBP under depolarizing noise. However, under pure $X$, $Z$ and $Y$ noise, the message update and hard decision of FDBP is mathematically reduced to those of PDBP, which means the restraint condition is not working in this case. Thus, they have almost the same error-correcting performance under pure $X$, $Z$ and $Y$ noise.

\subsection {Comparison between sum-product-based SBP, PDBP and FDBP}
\label{4.2}
Fig. \ref{fig:ps-based Comparison between SBP, PDBP and FDBP on XZZX surface code pure noise} show the LER as a function of physical qubit error rate of sum-product-based SBP, PDBP and FDBP on XZZX surface code with different lattice size $L$ under pure $X$, $Z$ and $Y$ noise. Fig. \ref{fig:ps-based Comparison between SBP, PDBP and FDBP on XZZX surface code pure noise} (a)$\sim$(f) show that under pure X and $Z$ noise, the curves of sum-product-based SBP, PDBP and FDBP also almost coincide. Under pure $X$ and $Z$ noise, the code-capacity thresholds of sum-product-based SBP, PDBP and FDBP on XZZX surface code are all about $50\%$. Under pure $Y$ noise, within physical qubit error rate regime $0.15\sim0.3$, we observe that the larger the lattice size $L$ is, the higher the LER of sum-product-based PDBP and FDBP is, which is quite abnormal. This is because that at the step of \textbf{vertical update} of sum-product-based PDBP and FDBP, we may face with the situation that an infinite number minus another infinite number. In this case, computed result is ``\textbf{Not a Number}'' in computer. Thus, we should reset the computed result to $0$, which leads to a decrease in computational accuracy. This situation is more likely to occur under pure $Y$ noise  within physical qubit error rate regime $0.15\sim0.3$ in our simulations. Nonetheless, as shown in Fig. \ref{fig:ps-based Comparison between SBP, PDBP and FDBP on XZZX surface code pure noise} (g)$\sim$(i), their LER is much lower than that of SBP. 

Similarly, Fig. \ref{fig:ps-based Comparison between SBP, PDBP and FDBP on planar surface code pure noise} (a)$\sim$(c) and (d)$\sim$(f) show the code-capacity thresholds of sum-product-based SBP, PDBP and FDBP on planar surface code under pure $X$ and $Z$ noise are all about $10\%$, and as shown in  Fig. \ref{fig:ps-based Comparison between SBP, PDBP and FDBP on planar surface code pure noise} (g)$\sim$(i), the LER of sum-product-based PDBP and FDBP is also much lower than that of SBP.

The above results once again support that in the decoupling representation of Pauli errors decoders tend to have more balanced capability to decode Pauli $X$, $Y$ and $Z$ errors.

\begin{figure*}[htbp]
	\centering
	\begin{minipage}{1\linewidth}
		\centering
		\includegraphics[width=1\linewidth]{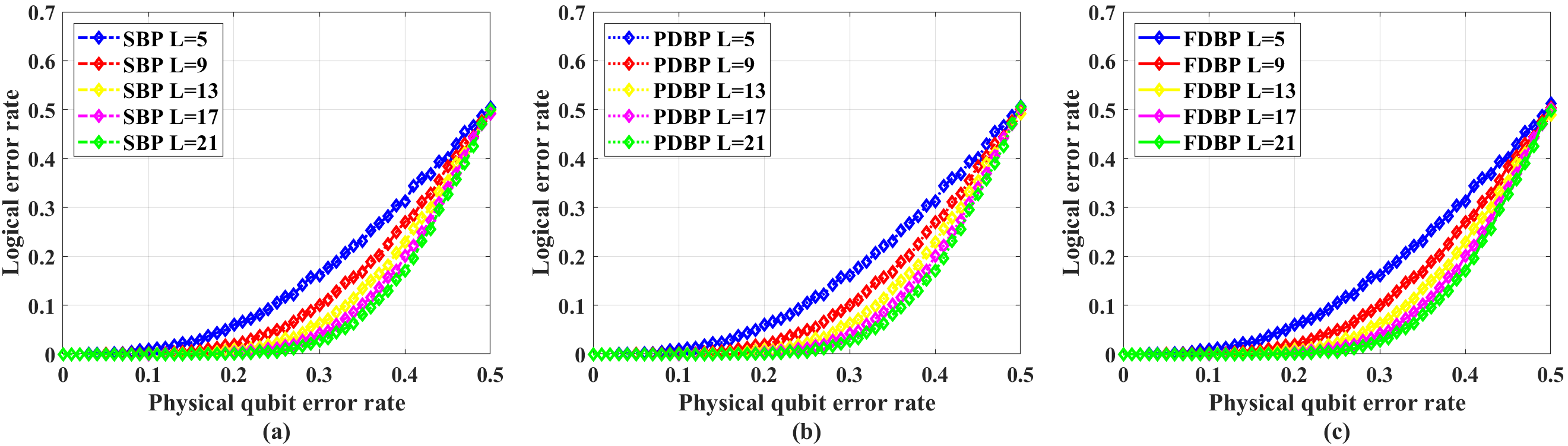}
	\end{minipage}
	\begin{minipage}{1\linewidth}
		\centering
		\includegraphics[width=1\linewidth]{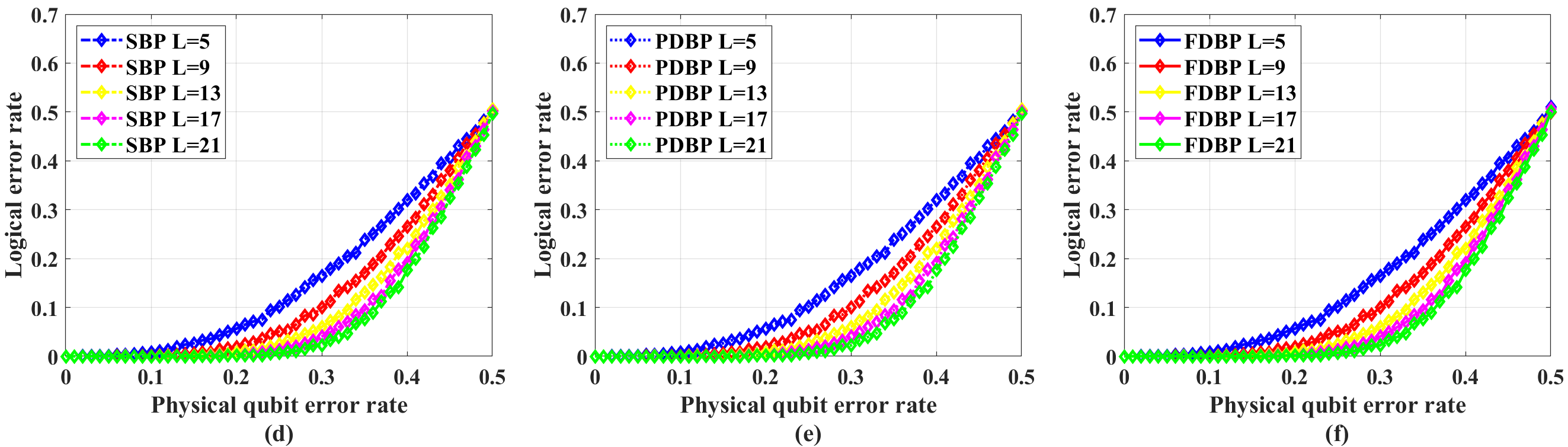}
	\end{minipage}
	\begin{minipage}{1\linewidth}
		\centering
		\includegraphics[width=1\linewidth]{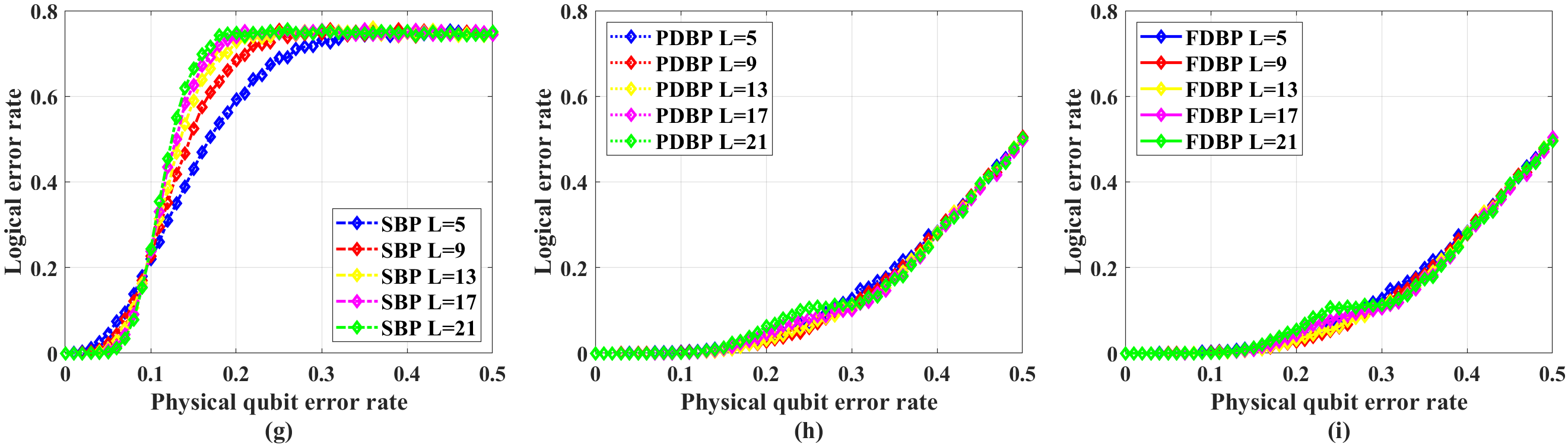}
	\end{minipage}
	\caption{\textbf{Logical error rate as a function of physical qubit error rate and code-capacity thresholds of sum-product-based SBP, PDBP and FDBP on XZZX surface code of different lattice size $L$ under pure $X$, $Z$ and $Y$ noise.}(a) Sum-product-based SBP on XZZX surface codes under pure $X$ noise. (b) Sum-product-based PDBP on XZZX surface codes under pure $X$ noise. (c) Sum-product-based FDBP on XZZX surface codes under pure $X$ noise. (d) Sum-product-based SBP on XZZX surface codes under pure $Z$ noise. (e) Sum-product-based PDBP on XZZX surface codes under pure $Z$ noise. (f) Sum-product-based FDBP on XZZX surface codes under pure $Z$ noise. (g) Sum-product-based SBP on XZZX surface codes under pure $Y$ noise. (h) Sum-product-based PDBP XZZX surface codes under pure $Y$ noise. (i) Sum-product-based FDBP on XZZX surface codes under pure $Y$ noise.}
	\label{fig:ps-based Comparison between SBP, PDBP and FDBP on XZZX surface code pure noise}
\end{figure*}

\begin{figure*}[htbp]
	\centering
	
	\begin{minipage}{1\linewidth}
		\centering
		\includegraphics[width=1\linewidth]{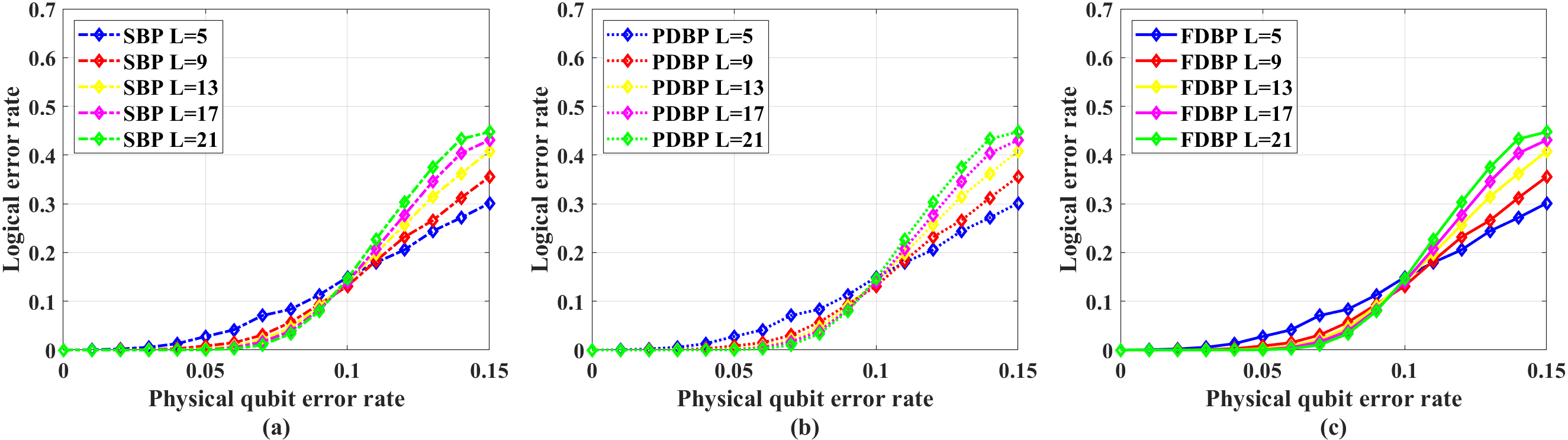}
	\end{minipage}
	\begin{minipage}{1\linewidth}
		\centering
		\includegraphics[width=1\linewidth]{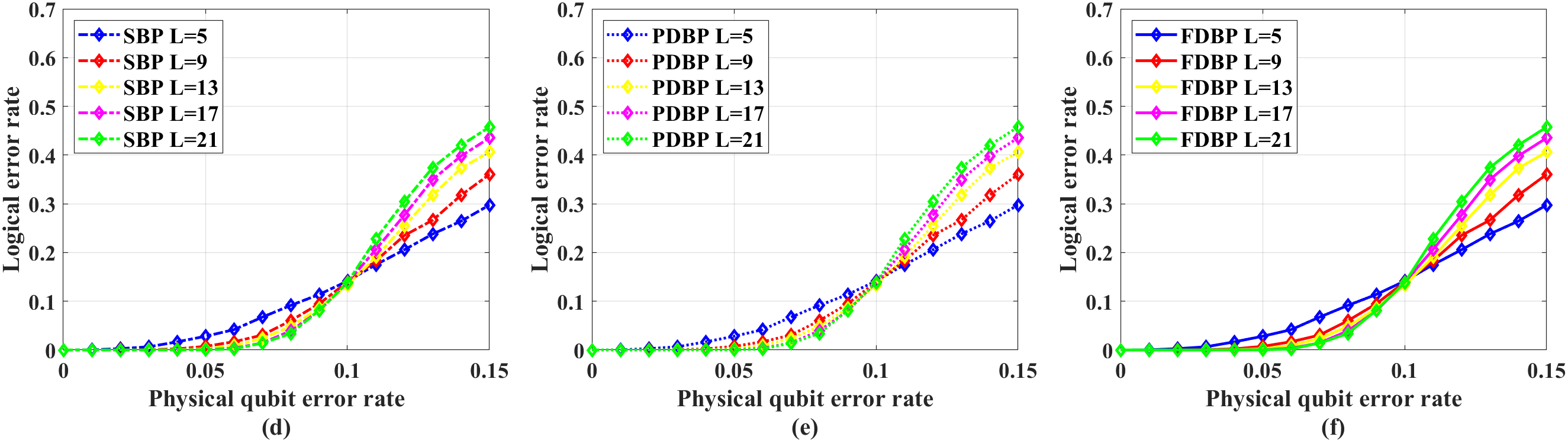}
	\end{minipage}
	\begin{minipage}{1\linewidth}
		\centering
		\includegraphics[width=1\linewidth]{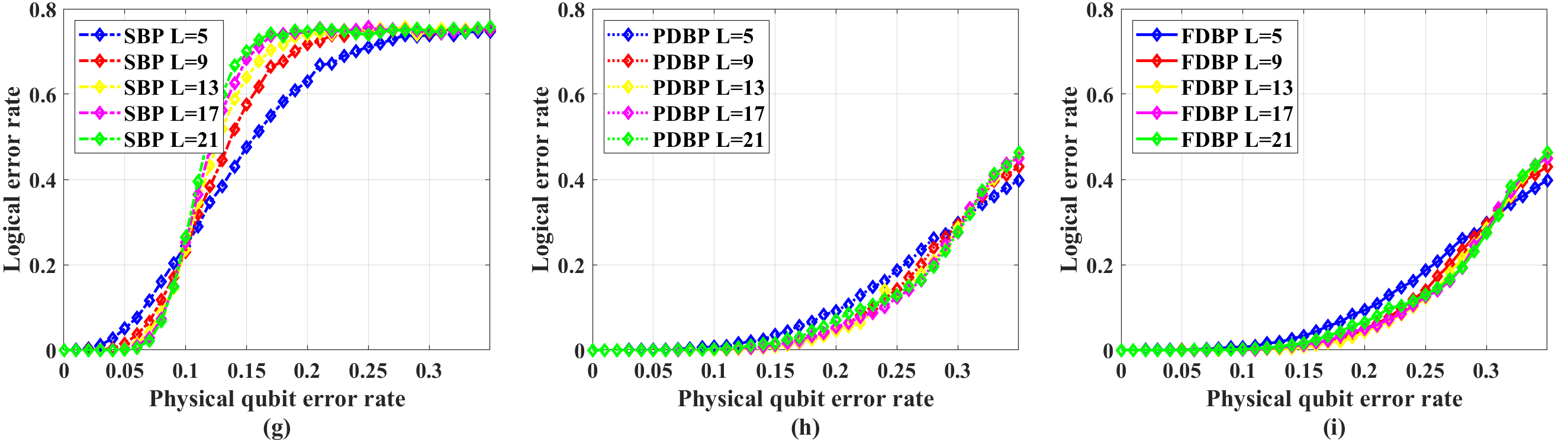}
	\end{minipage}
	\caption{\textbf{Logical error rate as a function of physical qubit error rate and code-capacity thresholds of sum-product-based SBP, PDBP and FDBP on planar surface code of different lattice size $L$ under pure $X$, $Z$ and $Y$ noise.}(a) Sum-product-based SBP on planar surface codes under pure $X$ noise. (b) Sum-product-based PDBP on planar surface codes under pure $X$ noise. (c) Sum-product-based FDBP on planar surface codes under pure $X$ noise. (d) Sum-product-based SBP on planar surface codes under pure $Z$ noise. (e) Sum-product-based PDBP on planar surface codes under pure $Z$ noise. (f) Sum-product-based FDBP on planar surface codes under pure $Z$ noise. (g) Sum-product-based SBP on planar surface codes under pure $Y$ noise. (h) Sum-product-based PDBP planar surface codes under pure $Y$ noise. (i) Sum-product-based FDBP on planar surface codes under pure $Y$ noise.}
	\label{fig:ps-based Comparison between SBP, PDBP and FDBP on planar surface code pure noise}
\end{figure*}

As shown in Fig. \ref{fig:ps-based Comparison between SBP, PDBP and FDBP on XZZX surface code_jpg} (a)$\sim$(c), the code-capacity thresholds of sum-product-based SBP, PDBP and FDBP on XZZX surface code under depolarizing noise are $14.2\%$, $16\%$ and $16.3\%$ , respectively. And Fig. \ref{fig:ps-based Comparison between SBP, PDBP and FDBP on planar surface code_jpg} (a)$\sim$(c) show the code-capacity thresholds of min-sum-based SBP, PDBP and FDBP on planar surface code under depolarizing noise are about $15\%$, $16\%$ and $16.8\%$, respectively. Sum-product-based FDBP has the highest code-capacity threshold. In order to further compare their difference, we select three curves corresponding to different lattice size $L=5$, $13$ and $21$ from Fig. \ref{fig:ps-based Comparison between SBP, PDBP and FDBP on XZZX surface code_jpg} (a)$\sim$(c) respectively and obtain Fig. \ref{fig:ps-based Comparison between SBP, PDBP and FDBP on XZZX surface code_jpg} (d)$\sim$(f). Doing similarly for Fig. \ref{fig:ps-based Comparison between SBP, PDBP and FDBP on planar surface code_jpg} (a)$\sim$(c) and we obatin Fig. \ref{fig:ps-based Comparison between SBP, PDBP and FDBP on planar surface code_jpg} (d)$\sim$(f). We can see that, the LER of sum-product-based FDBP is the lowest when $L>5$. Moreover, the larger the lattice size is, the larger the difference in LER between sum-product-based FDBP and another two algorithms is. Moreover, as shown in  Fig. \ref{fig:ps-based Comparison between SBP, PDBP and FDBP on XZZX surface code_jpg} (g)$\sim$(i) and Fig. \ref{fig:ps-based Comparison between SBP, PDBP and FDBP on planar surface code_jpg} (g)$\sim$(i), the average number of iterations of sum-product-based FDBP is the least.

\begin{figure*}[htbp]
	\centering
	\begin{minipage}{1\linewidth}
		\centering
		\includegraphics[width=1\linewidth]{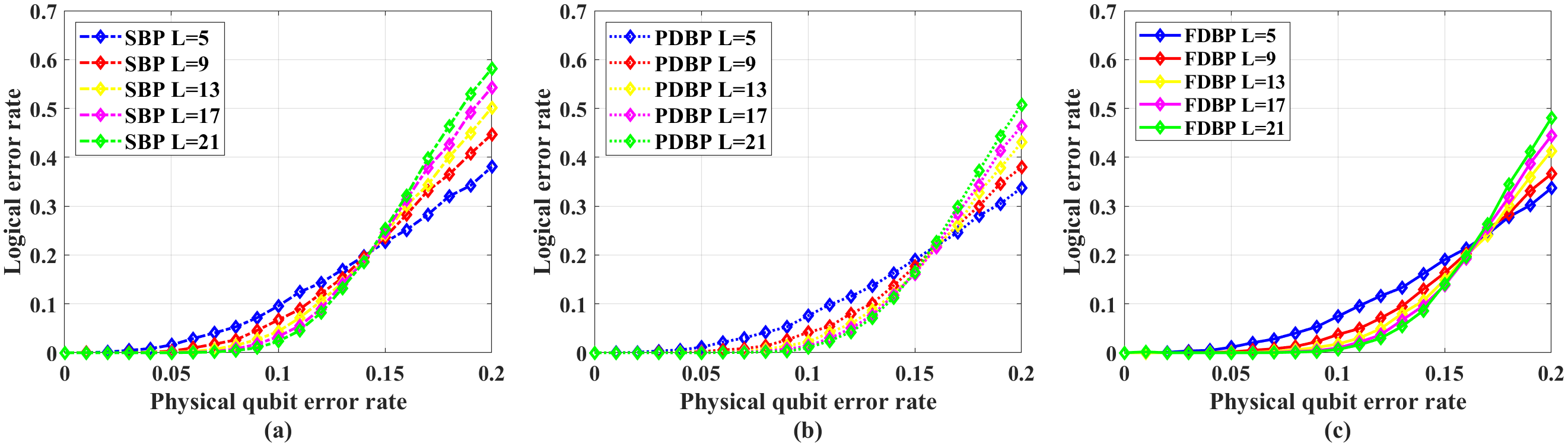}
	\end{minipage}
	\begin{minipage}{1\linewidth}
		\centering
		\includegraphics[width=1\linewidth]{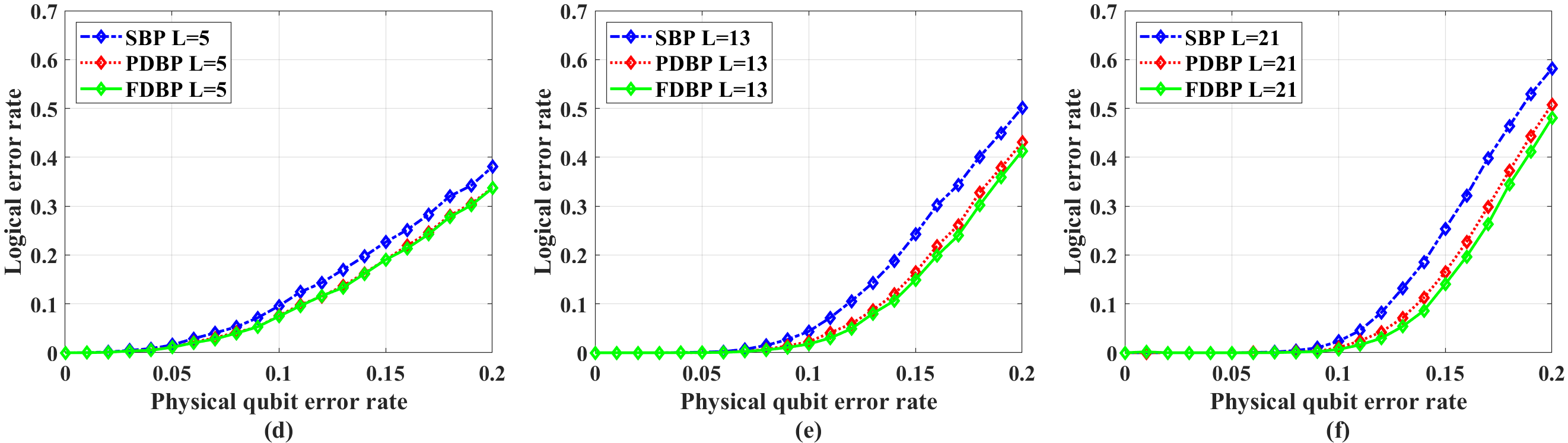}
	\end{minipage}
	\begin{minipage}{1\linewidth}
		\centering
		\includegraphics[width=1\linewidth]{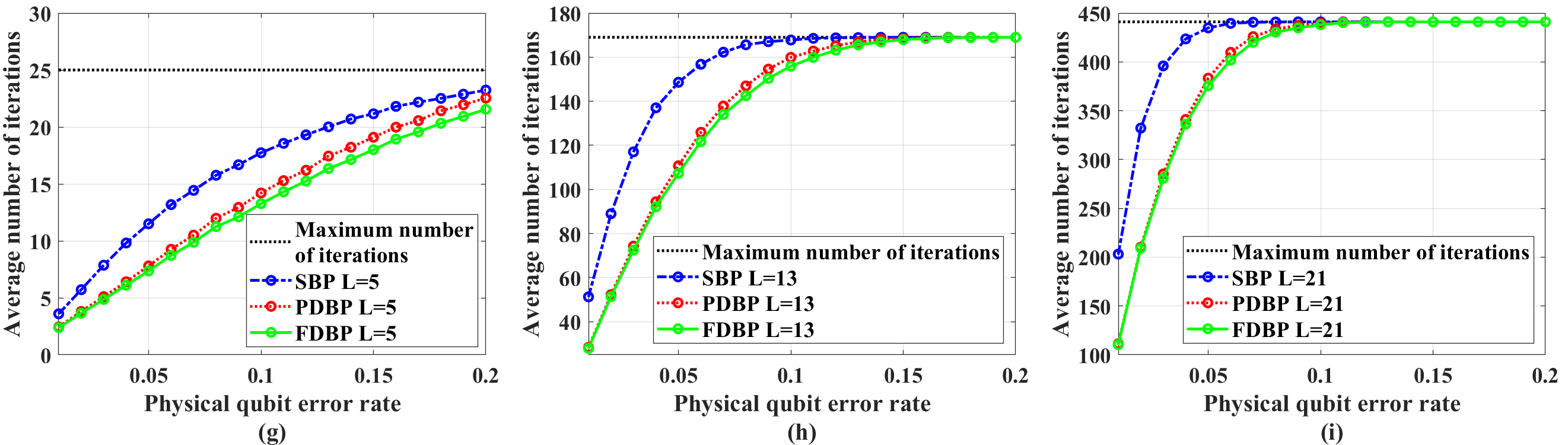}
	\end{minipage}
	\caption{Error-correcting performances of sum-product-based (a) SBP, (b) PDBP and (c) FDBP on XZZX surface code with different lattice size $L$ under depolarizing noise. We observe the code-capacity thresholds of sum-product-based SBP, PDBP and FDBP on XZZX surface code under depolarizing noise are $13.7\%$, $13.8\%$ and $16\%$, respectively. Comparison of LER of sum-product-based SBP, PDBP and FDBP on XZZX surface code with Lattice size (d) $L=5$, (e) $L=13$ and (e) $L=21$. The LER of FDBP is lowest when $L>5$ and the larger the lattice size is, the larger the difference in logical error rate between min-sum-based FDBP and another two algorithms is. The average number of iterations of sum-product-based SBP, PDBP and FDBP on XZZX surface code with Lattice size (g) $L=5$, (h) $L=13$ and (i) $L=21$. The maximum number of iterations is the code length $N=L^2$.}
	\label{fig:ps-based Comparison between SBP, PDBP and FDBP on XZZX surface code_jpg}
\end{figure*}

\begin{figure*}[htbp]
	\centering
	\begin{minipage}{1\linewidth}
		\centering
		\includegraphics[width=1\linewidth]{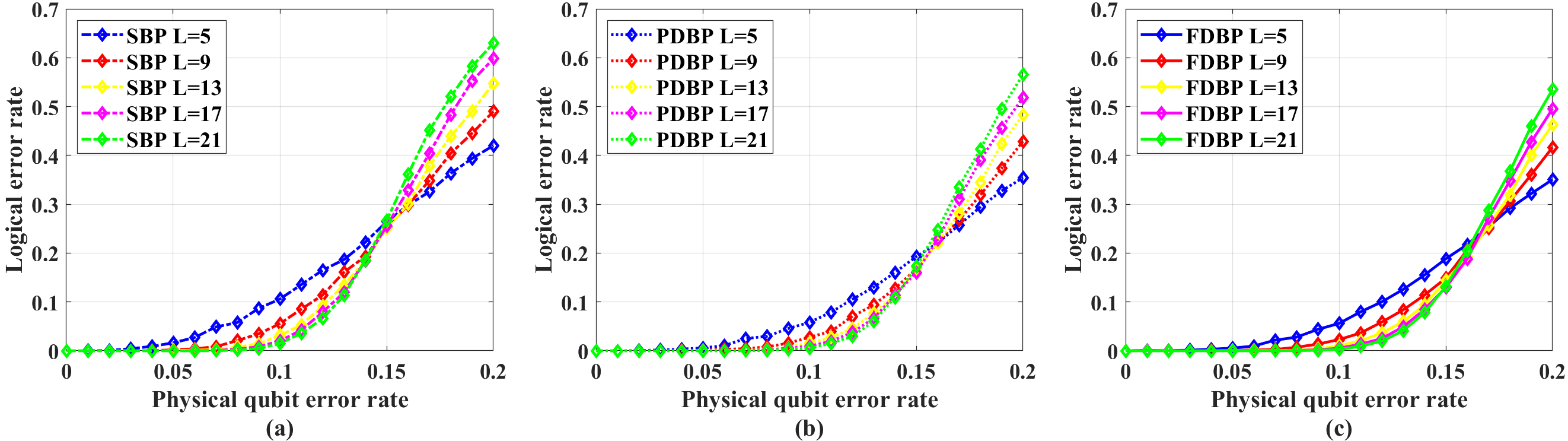}
	\end{minipage}
	\begin{minipage}{1\linewidth}
		\centering
		\includegraphics[width=1\linewidth]{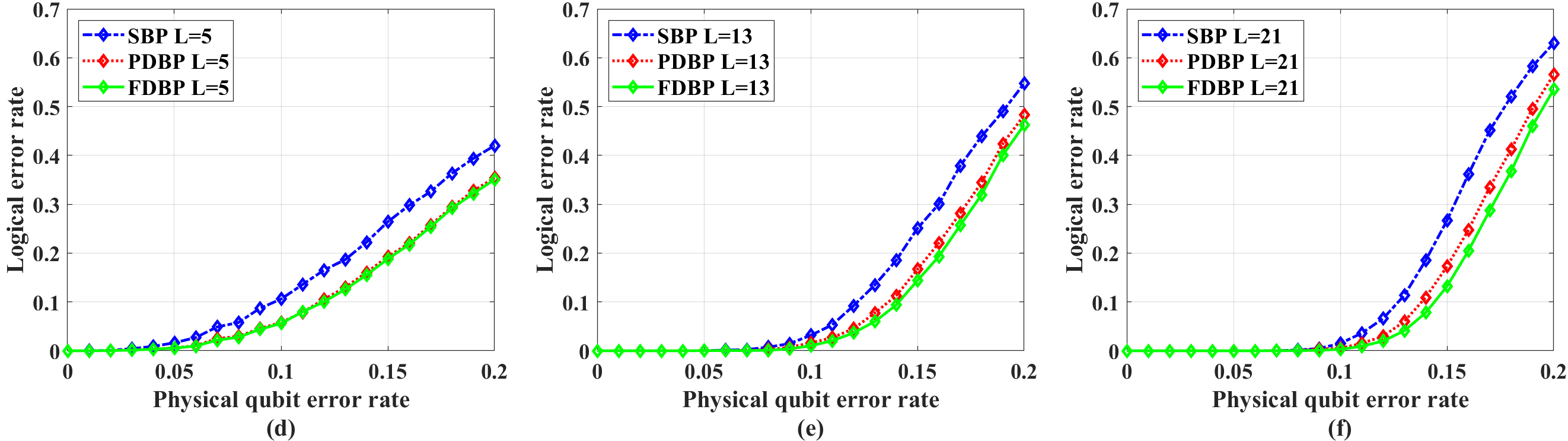}
	\end{minipage}
	\begin{minipage}{1\linewidth}
		\centering
		\includegraphics[width=1\linewidth]{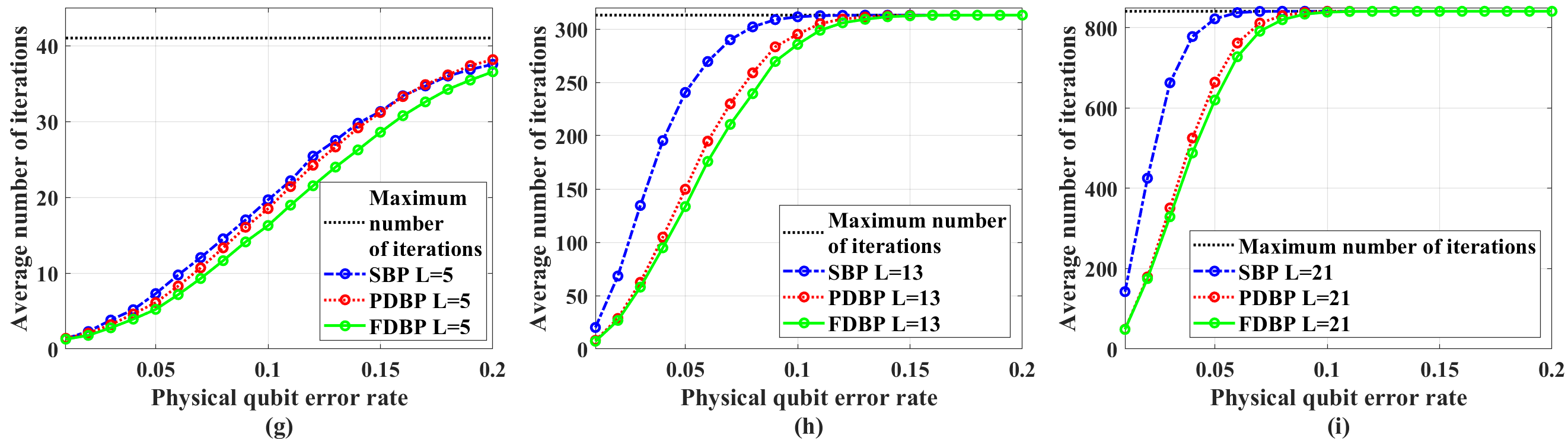}
	\end{minipage}
	\caption{Error-correcting performances of sum-product-based (a) SBP, (b) PDBP and (c) FDBP on planar surface code with different lattice size $L$ under depolarizing noise. We observe the code-capacity thresholds of sum-product-based SBP, PDBP and FDBP on planar surface code under depolarizing noise are $15\%$, $16\%$ and $16.8\%$, respectively. Comparison of LER of sum-product-based SBP, PDBP and FDBP on planar surface code with Lattice size (d) $L=5$, (e) $L=13$ and (e) $L=21$. The LER of FDBP is lowest and the larger the lattice size is, the larger the difference in logical error rate between sum-product-based FDBP and another two algorithms is. The average number of iterations of sum-product-based SBP, PDBP and FDBP on planar surface code with Lattice size (g) $L=5$, (h) $L=13$ and (i) $L=21$. The maximum number of iterations is the code length $N=2L^2-2L+1$. }
	\label{fig:ps-based Comparison between SBP, PDBP and FDBP on planar surface code_jpg}
\end{figure*}

\subsection {Comparison between min-sum-based and sum-product-based message update  algorithms}
\label{4.3}
In theory, min-sum-based message update algorithm is an approximation of sum-product-based message update algorithm \cite{fossorier1999reduced}, which can be reflected by their differences of LER and code-capacity threshold.
For LER, Fig. \ref{fig:ms_ps_compare_FDBP_PDBP_XZZX_depolarizing} (a)$\sim$(c) and Fig. \ref{fig:ms_ps_compare_FDBP_PDBP_surface_depolarizing} (a)$\sim$(c) show that under depolarizing noise, the LER of sum-product-based FDBP is lower than that of min-sum-based FDBP on XZZX surface code and planar surface code. Similarly, Fig. \ref{fig:ms_ps_compare_FDBP_PDBP_XZZX_depolarizing} (d)$\sim$(f) and Fig. \ref{fig:ms_ps_compare_FDBP_PDBP_surface_depolarizing} (d)$\sim$(f) show that the LER of sum-product-based PDBP is also lower than that of min-sum-based PDBP in the same case. For code-capacity threshold, under depolarizing noise, Fig. \ref{fig:ps-based Comparison between SBP, PDBP and FDBP on XZZX surface code_jpg} (b) and (c) show the code-capacity thresholds of sum-product-based PDBP and FDBP on XZZX surface code are $16\%$ and $16.3\%$, respectively, which are higher than those of min-sum-based PDBP and FDBP, $13.8\%$ and $16\%$, as shown in Fig. \ref{fig:ms-based Comparison between SBP, PDBP and FDBP on XZZX surface code_jpg} (b) and (c), respectively. Similarly, Fig. \ref{fig:ps-based Comparison between SBP, PDBP and FDBP on planar surface code_jpg} (b) and (c) show the code-capacity thresholds of sum-product-based PDBP and FDBP on planar surface code are $16\%$ and $16.8\%$, respectively, which are higher than those of min-sum-based PDBP and FDBP, $13.5\%$ and $15.6\%$, as shown in Fig. \ref{fig:ms-based Comparison between SBP, PDBP and FDBP on planar surface code_jpg} (b) and (c), respectively.

\begin{figure*}[htbp]
	\centering
	\includegraphics[width=1\textwidth]{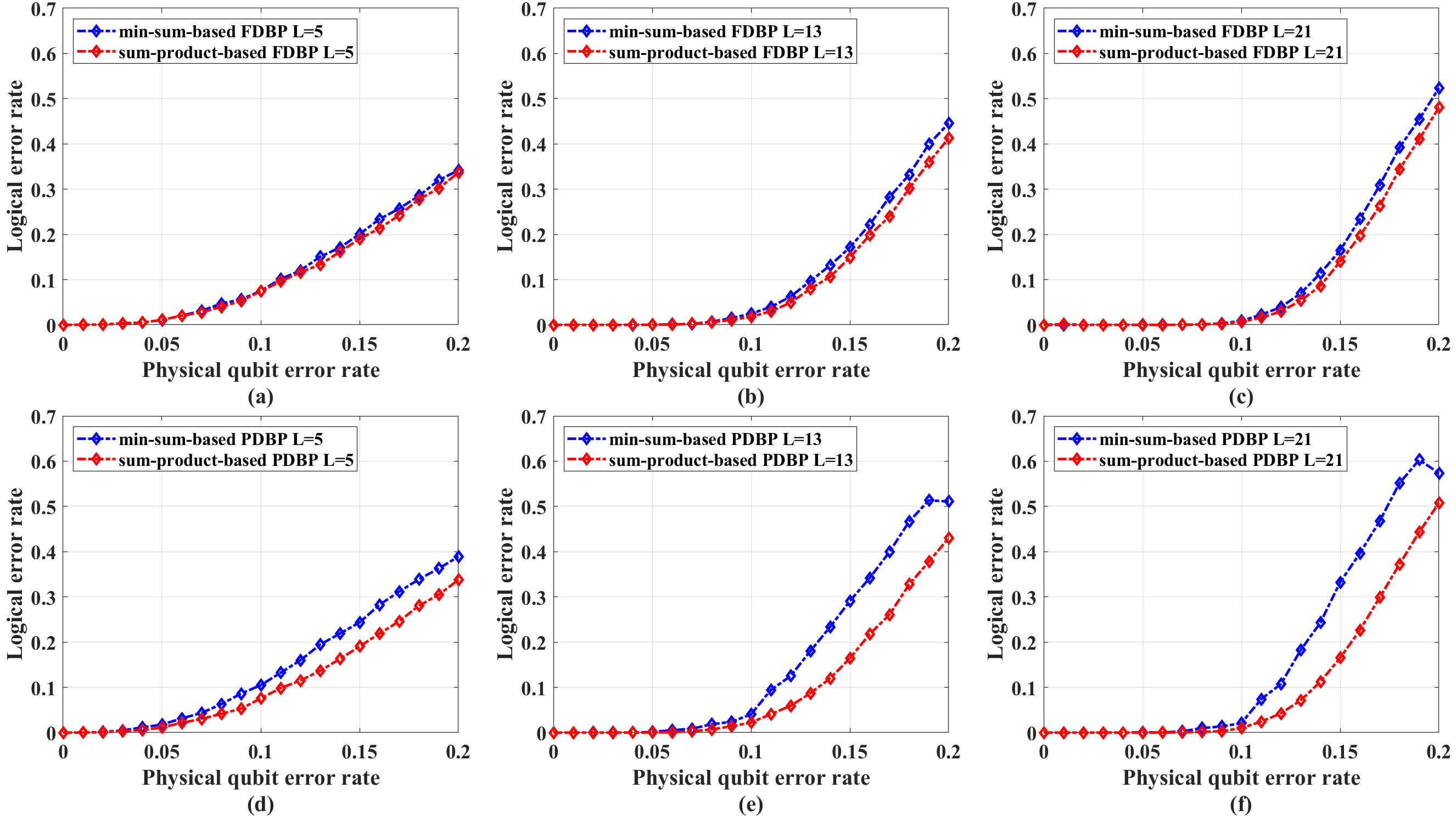}
	\caption{Comparison of LER between min-sum-based FDBP and sum-product-based FDBP on XZZX surface code with lattice size (a) $L=5$, (b) $L=13$ and (c) $L=21$ under depolarizing noise. Comparison of LER between min-sum-based PDBP and sum-product-based PDBP on XZZX surface code with lattice size (d) $L=5$, (e) $L=13$ and (f) $L=21$ under depolarizing noise.}
	\label{fig:ms_ps_compare_FDBP_PDBP_XZZX_depolarizing}
\end{figure*}

\begin{figure*}[htbp]
	\centering
	\includegraphics[width=1\textwidth]{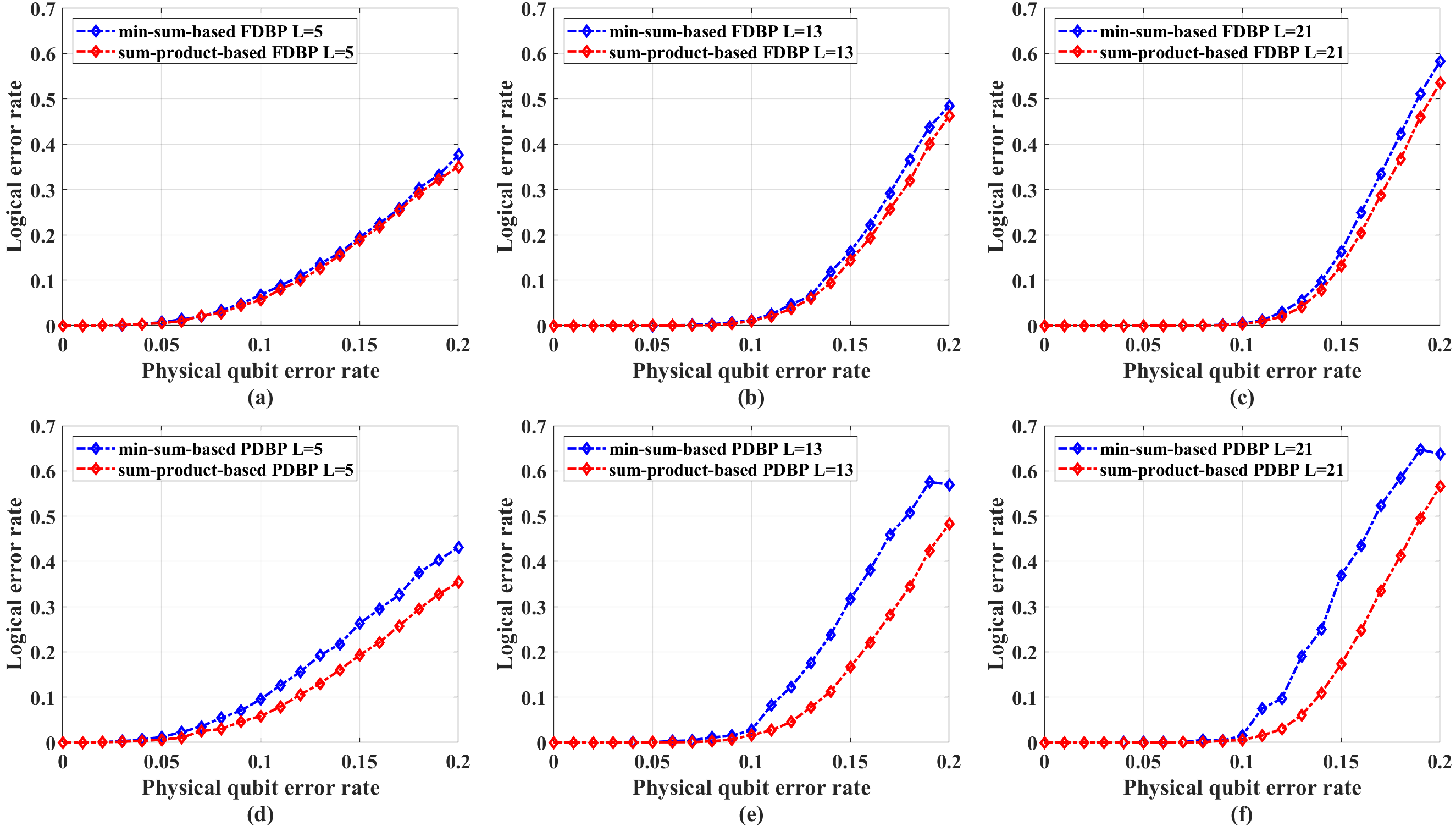}
	\caption{Comparison of LER between min-sum-based FDBP and sum-product-based FDBP on planar surface code with lattice size (a) $L=5$, (b) $L=13$ and (c) $L=21$ under depolarizing noise. Comparison of LER between min-sum-based PDBP and sum-product-based PDBP on planar surface code with lattice size (d) $L=5$, (e) $L=13$ and (f) $L=21$ under depolarizing noise.}
	\label{fig:ms_ps_compare_FDBP_PDBP_surface_depolarizing}
\end{figure*}

\section {Conclusion}
\label{5}
We propose a new method called decoupling representation to represent Pauli operators acting on $n$ qubits as binary vectors with size of $3n$. With decoupling representation, we propose PDBP and FDBP decoding algorithms. These algorithms can apply to both non-CSS and CSS codes to deal with the correlations between the $X$ part and the $Z$ part of the vectors in symplectic representation, which are introduced by Pauli $Y$ errors. Simulation results on XZZX surface code and planar surface code show that PDBP and FDBP greatly outperform SBP under pure $Y$ noise channel and depolarizing noise channel, while under pure $X$ and $Z$ noise channel their performances are the almost the same. The code-capacity thresholds of min-sum-based PDBP and FDBP on XZZX surface codes under pure Pauli $X$, $Z$ and $Y$ noise channels are about $50\%$, $50\%$ and $40\%$, respectively, while the code-capacity threshold of min-sum-based SBP are $50\%$, $50\%$ and $9\%$, respectively. And the code-capacity thresholds of min-sum-based PDBP and FDBP on planar surface codes under pur $X$, $Z$ and $Y$ noise channels are about $9\%$, $9\%$ and $25\%$, respectively, while the code-capacity thresholds of min-sum-based SBP are all about $9\%$. For sum-product-based SBP, PDBP and FDBP, under pure $X$ and $Z$ noise, thier code-capacity thresholds on XZZX surface code and planar surface code are all $50\%$ and $10\%$, respectively. For pure $Y$ noise, due to the decrease in computational accuracy caused by the case that an infinite number minus another infinite number, we can not observe explicit code-capacity thresholds of sum-product-based PDBP and FDBP, but their LER is much lower than that of sum-product-based SBP. These results support that in decoupling representation decoders tend to have more balanced capability to decode Pauli $X$, $Y$ and $Z$ errors.

More surprisingly, FDBP performs better than PDBP and SBP under depolarizing noise channel. For XZZX surface code, the code-capacity thresholds of min-sum-based SBP, PDBP and FDBP under depolarizing noise channel are $13.7\%$, $13.8\%$ and $16\%$ respectively, and those of sum-product-based SBP, PDBP and FDBP are $14.2\%$, $16\%$ and $16.3\%$ respectively.
For planar surface code, the code-capacity thresholds of min-sum-based SBP, PDBP and FDBP under depolarizing noise channel are $14\%$, $13.5\%$ and $15.6\%$ respectively, and those of sum-product-based SBP, PDBP and FDBP are $15\%$, $16\%$ and $16.8\%$ respectively.

However, there are still some major open questions. The first is that it’s still unclear how to reveal the reason for PDBP and FDBP outperform SBP in depolarizing noise channel from the point of the structure of Tanner graphs. Since in depolarizing noise channel, the PDBP and FDBP’s Tanner graphs might have more loops than SBP. Empirically, more loops might result in worse performance of BP. Hence, we can’t analysis the reason by simply comparing the number of loops as in pure $Y$ noise channel. What's more, though the time complexity of PDBP and FDBP is $O(N)$, the time complexity of order$-0$ OSD is $O(N^3)$. In the worst case where all iterations of BP have been carried out, the whole time complexity of BP combined with order$-0$ OSD is $O(N+N^3)$. Thus the second open question is whether PDBP and FDBP can achieve a good perfermance without OSD and how to reduce the time cost of OSD. The third major open question is that it’s still unclear how to implement PDBP and FDBP when there is existing measurement noise.

\section*{Acknowledgements}
We acknowledge useful discussions with Jiahan Chen, Zicheng Wang and Fusheng Yang.

%

\paragraph{Author contributions}
Z.Z. Yi and Z.P. Liang discussed and worked jointly on this work. K.X. Zhong completed the code of C++ version of FDBP algorithm. All authors helped to write the article.

\paragraph{Funding information}
This work was supported by the Colleges and Universities Stable Support Project of Shenzhen, China (No.GXWD20220817164856008), the Colleges and Universities Stable Support Project of Shenzhen, China (No.GXWD20220811170225001) and Harbin Institute of Technology, Shenzhen - SpinQ quantum information Joint Research Center Project (No.HITSZ20230111).

\begin{appendix}
\numberwithin{equation}{section}

\section{Mathematical derivation of the message update rules of sum-product-based FDBP}
\label{Appendix A}
In this section, we give the mathematical derivation of the message update rules of sum-product-based FDBP.

Given a syndrome $\textbf{s}=(s_1,\cdots,s_{n-k})$ of a $[[n,k]]$ QLDPC code with decoupled parity-check matrix $H_d$, the FDBP is to calculate the log-likelihood ratio of each bit of $\textbf{e}=(e_1,\cdots,e_{3n})$ conditioned on $\textbf{s}$, namely,
\begin{equation}
	\ln \frac{p\left(\hat{e}_i=0|\textbf{s}\right)}{p\left(\hat{e}_i=1|\textbf{s}\right)}=\ln \frac{p\left(\hat{e}_i=0,\textbf{s}\right)/p\left(\textbf{s}\right)}{p\left(\hat{e}_i=1,\textbf{s}\right)/p\left(\textbf{s}\right)}=\ln \frac{p\left(\hat{e}_i=0,\textbf{s}\right)}{p\left(\hat{e}_i=1,\textbf{s}\right)}
\end{equation}
for all $1\le i\le 3n$, where the following restraint condition should be taken into account.
\begin{equation}
	\label{appendix constraint condition}
	\hat{e}_i+\hat{e}_{(i+n)\ mod\ 3n}+\hat{e}_{(i+2n)\ mod\ 3n}\leq1
\end{equation}

We first calculate $p\left(\hat{e}_i=0,\textbf{s}\right)$. Using the Total Probability Theorem, we have
\begin{equation}
	p\left(\hat{e}_i=0,\textbf{s}\right)=\sum_{\substack{\hat{e}\in GF(2)^{3n}:\hat{e}_i=0\\ H_d\cdot \hat{\textbf{e}}=\textbf{s}}}p(\hat{\textbf{e}})
\end{equation}
We use $H_d^j$ to denote the $j$th row of $H_d$, $\mathcal{N}(j)$ to denote the indices set of nonzero elements of $H_d^j$ and $\mathcal{N}(i)$ to denote the indices set of nonzero elements of the $i$th column of $H_d$. Supposing each check equation is independent with each other, we obtain
\begin{equation}
	\begin{aligned}
		&p\left(\hat{e}_i=0,\textbf{s}\right)=\prod_{j\in \mathcal{N}(i)}\sum_{\substack{\hat{\textbf{e}}\in GF(2)^{3n}:\hat{e}_i=0\\ H_d^j\cdot \hat{\textbf{e}}=s_j}}p(\hat{\textbf{e}})\\
		&=\prod_{j\in \mathcal{N}(i)}\left[\sum_{\substack{\hat{\textbf{e}}\in GF(2)^{3n}:\hat{e}_i=0\\ H_d^j\cdot \hat{\textbf{e}}=s_j}}\prod_{i^{\prime}\in \mathcal{N}(j)/i}p(\hat{e}_{i^{\prime}})\right]p(\hat{e}_{i}=0)
	\end{aligned}
\end{equation}
For $p\left(\hat{e}_i=1,\textbf{s}\right)$, we should additionally take Eq. (\ref{appendix constraint condition}) into account and obtain
\begin{equation}
	\begin{aligned}
		&p\left(\hat{e}_i=0,\textbf{s}\right)=\prod_{j\in \mathcal{N}(i)}\sum_{\substack{\hat{\textbf{e}}\in GF(2)^{3n}:\hat{e}_i=0\\ H_d^j\cdot \hat{\textbf{e}}=s_j}}p(\hat{\textbf{e}})\\
		&=\prod_{j\in \mathcal{N}(i)}\left[\sum_{\substack{\hat{\textbf{e}}\in GF(2)^{3n}:\hat{e}_i=1,\hat{e}_{i_s}=0\\ H_d^j\cdot \hat{\textbf{e}}=s_j\oplus 1}}\prod_{i^{\prime}\in \mathcal{N}(j)/\{i,i_s\}}p(\hat{e}_{i^{\prime}})\right]\\
		&\times p(\hat{e}_{i}=1)p(\hat{e}_{i_s}=0)
	\end{aligned}
\end{equation}
Thus, we have
\begin{equation}
	\begin{aligned}
		&\ln\frac{p\left(\hat{e}_i=0,\textbf{s}\right)}{p\left(\hat{e}_i=1,\textbf{s}\right)}=\ln \frac{p\left(\hat{e}_i=0\right)}{p\left(\hat{e}_i=1\right)p\left(\hat{e}_{i_s}=0\right)}+\ln\prod_{j\in\mathcal{N}(i)}\frac{\sum_{\substack{\hat{\textbf{e}}\in GF(2)^{3n}:\hat{e}_i=0\\ H_d^j\cdot \hat{\textbf{e}}=s_j}}\prod_{i^{\prime}\in \mathcal{N}(j)/i}p(\hat{e}_{i^{\prime}})}{\sum_{\substack{\hat{\textbf{e}}\in GF(2)^{3n}:\hat{e}_i=1,\hat{e}_{i_s}=0\\ H_d^j\cdot \hat{\textbf{e}}=s_j\oplus 1}}\prod_{i^{\prime}\in \mathcal{N}(j)/\{i,i_s\}}p(\hat{e}_{i^{\prime}})}
	\end{aligned}
\end{equation}
if $s_j=0$, we have
\begin{equation}
	\begin{aligned}
		\label{sj=0}
		&\ln\prod_{j\in\mathcal{N}(i)}\frac{\sum_{\substack{\hat{\textbf{e}}\in GF(2)^{3n}:\hat{e}_i=0\\ H_d^j\cdot \hat{\textbf{e}}=s_j}}\prod_{i^{\prime}\in \mathcal{N}(j)/i}p(\hat{e}_{i^{\prime}})}{\sum_{\substack{\hat{\textbf{e}}\in GF(2)^{3n}:\hat{e}_i=1,\hat{e}_{i_s}=0\\ H_d^j\cdot \hat{\textbf{e}}=s_j\oplus 1}}\prod_{i^{\prime}\in \mathcal{N}(j)/\{i,i_s\}}p(\hat{e}_{i^{\prime}})}\\
		&=\sum_{j\in\mathcal{N}(i)}\ln \frac{\sum_{\substack{\hat{\textbf{e}}\in GF(2)^{3n}:\hat{e}_i=0\\ H_d^j\cdot \hat{\textbf{e}}=s_j}}\prod_{i^{\prime}\in \mathcal{N}(j)/i}p(\hat{e}_{i^{\prime}})}{\sum_{\substack{\hat{\textbf{e}}\in GF(2)^{3n}:\hat{e}_i=1,\hat{e}_{i_s}=0\\ H_d^j\cdot \hat{\textbf{e}}=s_j\oplus 1}}\prod_{i^{\prime}\in \mathcal{N}(j)/\{i,i_s\}}p(\hat{e}_{i^{\prime}})}
	\end{aligned}
\end{equation}
The numerator of Eq. (\ref{sj=0}) means among the bits of $\hat{\textbf{e}}$ indexed $i^{\prime}\in \mathcal{N}(j)/i$ there are an even number of $1$s. Similarly, the denominator of Eq. (\ref{sj=0}) means among the bits of $\hat{\textbf{e}}$ indexed $i^{\prime}\in \mathcal{N}(j)/\{i,i_s\}$ there are an odd number of $1$. Thus, we have, 
\begin{equation}
	\begin{aligned}
		\label{1sj=0}
		&\sum_{j\in\mathcal{N}(i)}\ln \frac{\sum_{\substack{\hat{\textbf{e}}\in GF(2)^{3n}:\hat{e}_i=0\\ H_d^j\cdot \hat{\textbf{e}}=s_j}}\prod_{i^{\prime}\in \mathcal{N}(j)/i}p(\hat{e}_{i^{\prime}})}{\sum_{\substack{\hat{\textbf{e}}\in GF(2)^{3n}:\hat{e}_i=1,\hat{e}_{i_s}=0\\ H_d^j\cdot \hat{\textbf{e}}=s_j\oplus 1}}\prod_{i^{\prime}\in \mathcal{N}(j)/\{i,i_s\}}p(\hat{e}_{i^{\prime}})}\\
		&=\sum_{j\in\mathcal{N}(i)}\ln \frac{1+\prod_{i^{\prime}\in \mathcal{N}(j)/i}\tanh (\frac{\ln \frac{p(\hat{e}_{i^{\prime}}=0)}{p(\hat{e}_{i^{\prime}}=1)}}{2})}{1-\prod_{i^{\prime}\in \mathcal{N}(j)/\{i,i_s\}}\tanh (\frac{\ln \frac{p(\hat{e}_{i^{\prime}}=0)}{p(\hat{e}_{i^{\prime}}=1)}}{2})}
	\end{aligned}
\end{equation}
Similarly, if $s_j=1$, we have
\begin{equation}
	\begin{aligned}
		\label{sj=1}
		&\ln\prod_{j\in\mathcal{N}(i)}\frac{\sum_{\substack{\hat{\textbf{e}}\in GF(2)^{3n}:\hat{e}_i=0\\ H_d^j\cdot \hat{\textbf{e}}=s_j}}\prod_{i^{\prime}\in \mathcal{N}(j)/i}p(\hat{e}_{i^{\prime}})}{\sum_{\substack{\hat{\textbf{e}}\in GF(2)^{3n}:\hat{e}_i=1,\hat{e}_{i_s}=0\\ H_d^j\cdot \hat{\textbf{e}}=s_j\oplus 1}}\prod_{i^{\prime}\in \mathcal{N}(j)/\{i,i_s\}}p(\hat{e}_{i^{\prime}})}\\
		&=\sum_{j\in\mathcal{N}(i)}\ln \frac{\sum_{\substack{\hat{\textbf{e}}\in GF(2)^{3n}:\hat{e}_i=0\\ H_d^j\cdot \hat{\textbf{e}}=1}}\prod_{i^{\prime}\in \mathcal{N}(j)/i}p(\hat{e}_{i^{\prime}})}{\sum_{\substack{\hat{\textbf{e}}\in GF(2)^{3n}:\hat{e}_i=1,\hat{e}_{i_s}=0\\ H_d^j\cdot \hat{\textbf{e}}=0}}\prod_{i^{\prime}\in \mathcal{N}(j)/\{i,i_s\}}p(\hat{e}_{i^{\prime}})}\\
		&=\sum_{j\in\mathcal{N}(i)}\ln \frac{1-\prod_{i^{\prime}\in \mathcal{N}(j)/i}\tanh (\frac{\ln \frac{p(\hat{e}_{i^{\prime}}=0)}{p(\hat{e}_{i^{\prime}}=1)}}{2})}{1+\prod_{i^{\prime}\in \mathcal{N}(j)/\{i,i_s\}}\tanh (\frac{\ln \frac{p(\hat{e}_{i^{\prime}}=0)}{p(\hat{e}_{i^{\prime}}=1)}}{2})}
	\end{aligned}
\end{equation}
Thus, we can obtain the message update rules of sum-product-based FDBP algorithm, as describing in \textbf{Algorithm} \ref{Sum-product-based FDBP}.

\section {Mathematical derivation of the message update rules of min-sum-based FDBP}
\label{Appendix B}
In this section, we give the mathematical derivation of the message update rules of min-sum-based FDBP.

According to \textbf{Algorithm} \ref{Sum-product-based FDBP}, when $s_j=0$, the parity-to-variable message $m_{c_j\Rightarrow v_i}$ is
\begin{equation}
	\begin{aligned}
		m_{c_j\Rightarrow v_i}= \ln \frac{1+\prod_{v_{i^{\prime}}\in \mathcal{N}(c_j)\slash v_i}\tanh \left(\frac{m_{v_{i^{\prime}}\Rightarrow c_j}}{2}\right)}{1-\prod_{v_{i^{\prime}}\in \mathcal{N}(c_j)\slash \{v_i,v_{i_s}\}} \tanh \left(\frac{m_{v_{i^{\prime}}\Rightarrow c_j}}{2}\right)}
	\end{aligned}
\end{equation}
According to \cite{fossorier1999reduced}, we have
\begin{equation}
	\begin{aligned}
		&m_{c_j\Rightarrow v_i}= \ln \frac{1+\prod_{v_{i^{\prime}}\in \mathcal{N}(c_j)\slash v_i}\tanh \left(\frac{m_{v_{i^{\prime}}\Rightarrow c_j}}{2}\right)}{1-\prod_{v_{i^{\prime}}\in \mathcal{N}(c_j)\slash v_i}\tanh \left(\frac{m_{v_{i^{\prime}}\Rightarrow c_j}}{2}\right)}+\ln \frac{1-\prod_{v_{i^{\prime}}\in \mathcal{N}(c_j)\slash v_i}\tanh \left(\frac{m_{v_{i^{\prime}}\Rightarrow c_j}}{2}\right)}{1-\prod_{v_{i^{\prime}}\in \mathcal{N}(c_j)\slash \{v_i,v_{i_s}\}} \tanh \left(\frac{m_{v_{i^{\prime}}\Rightarrow c_j}}{2}\right)}\\
		&\approx \left(\prod_{v_{i^{\prime}}\in \mathcal{N}(c_j)\slash v_i}sign(m_{v_{i^{\prime}}\Rightarrow c_j})\times \min_{v_{i^{\prime}}\in \mathcal{N}(c_j)\slash v_i}(\lvert m_{v_{i^{\prime}}\Rightarrow c_j} \rvert)\right)+\ln \frac{1-\prod_{v_{i^{\prime}}\in \mathcal{N}(c_j)\slash v_i}\tanh \left(\frac{m_{v_{i^{\prime}}\Rightarrow c_j}}{2}\right)}{1-\prod_{v_{i^{\prime}}\in \mathcal{N}(c_j)\slash \{v_i,v_{i_s}\}} \tanh \left(\frac{m_{v_{i^{\prime}}\Rightarrow c_j}}{2}\right)}
	\end{aligned}
\end{equation}
Similarly, when $s_j=1$, we have 
\begin{equation}
	\begin{aligned}
		&m_{c_j\Rightarrow v_i}= -\ln \frac{1+\prod_{v_{i^{\prime}}\in \mathcal{N}(c_j)\slash \{v_i,v_{i_s}\}}\tanh \left(\frac{m_{v_{i^{\prime}}\Rightarrow c_j}}{2}\right)}{1-\prod_{\mathcal{N}(c_j)\slash \{v_i,v_{i_s}\}}\tanh \left(\frac{m_{v_{i^{\prime}}\Rightarrow c_j}}{2}\right)}+\ln \frac{1-\prod_{v_{i^{\prime}}\in \mathcal{N}(c_j)\slash v_i}\tanh \left(\frac{m_{v_{i^{\prime}}\Rightarrow c_j}}{2}\right)}{1-\prod_{v_{i^{\prime}}\in \mathcal{N}(c_j)\slash \{v_i,v_{i_s}\}} \tanh \left(\frac{m_{v_{i^{\prime}}\Rightarrow c_j}}{2}\right)}\\
		&\approx \left(-1\times\prod_{v_{i^{\prime}}\in \mathcal{N}(c_j)\slash \{v_i,v_{i_s}\}}sign(m_{v_{i^{\prime}}\Rightarrow c_j})\times \min_{v_{i^{\prime}}\in \mathcal{N}(c_j)\slash v_i}(\lvert m_{v_{i^{\prime}}\Rightarrow c_j} \rvert)\right)\\
		&+\ln \frac{1-\prod_{v_{i^{\prime}}\in \mathcal{N}(c_j)\slash v_i}\tanh \left(\frac{m_{v_{i^{\prime}}\Rightarrow c_j}}{2}\right)}{1-\prod_{v_{i^{\prime}}\in \mathcal{N}(c_j)\slash \{v_i,v_{i_s}\}} \tanh \left(\frac{m_{v_{i^{\prime}}\Rightarrow c_j}}{2}\right)}
	\end{aligned}
\end{equation}

\end{appendix}





\bibliography{bibliography.bib}


\end{document}